\documentclass{lmcs} 
\pdfoutput=1

\usepackage[utf8]{inputenc}

\usepackage{lastpage}
\lmcsdoi{18}{1}{4}
\lmcsheading{}{\pageref{LastPage}}{}{}%
{May~25,~2021}{Jan.~07,~2022}{}

\keywords{Spatio-Temporal Logic, Monitoring, Cyber-Physical Systems}

\theoremstyle{plain} 

\theoremstyle{definition}
\newtheorem{definition}[thm]{Definition}
\newtheorem{example}[thm]{Example}
\theoremstyle{plain}
\newtheorem{proposition}[thm]{Proposition}
\newtheorem{lemma}[thm]{Lemma}
\newtheorem{theorem}[thm]{Theorem}
\theoremstyle{definition}

\usepackage{MnSymbol}
\usepackage{stmaryrd}
\usepackage{multicol}
\usepackage{amsthm}

\usepackage{algorithm}
\usepackage{algpseudocode}

\newcommand{\wfun}{\mathbf{W}}
\newcommand{\nextto}[3]{#1\stackrel{#2}{\mapsto}#3}
\newcommand{\route}{\tau}
\newcommand{\tsign}{\nu}
\newcommand{\pct}{\tilde{\tsign}}
\newcommand{\ssign}{\mathbf{s}}
\newcommand{\sts}{\sigma}
\newcommand{\pcsts}{\tilde{\sigma}}
\newcommand{\lserv}{\mathcal{S}}

\newcommand{\fmon}{\mathbf{m}}

\newcommand{\ev}[1]{\mathrm{F}_{#1}}
\newcommand{\glob}[1]{\mathrm{G}_{#1}}
\newcommand{\until}[1]{\mathrm{U}_{#1}}
\newcommand{\since}[1]{\mathrm{S}_{#1}}

\newcommand{\somewhere}[2]{\diamonddiamond_{#1}^{#2}}
\newcommand{\everywhere}[2]{\boxbox_{#1}^{#2}}
\newcommand{\surround}[2]{\circledcirc_{#1}^{#2}}
\newcommand{\reach}[2]{\mathcal{R}_{#1}^{#2}}
\newcommand{\escape}[2]{\mathcal{E}_{#1}^{#2}}

\newcommand{\coord}{ {\color{\green} coord } }
\newcommand{\router}{ {\color{\rouge} router} } 
\newcommand{\device}{ {\color{\blu} end\_dev } }

\definecolor{lavander}{cmyk}{0,0.48,0,0}
\definecolor{violet}{cmyk}{0.79,0.88,0,0}
\definecolor{burntorange}{cmyk}{0,0.52,1,0}

\def\blu{blue!45!black!70!}
\def\green{green!45!black!70!}

\def\rouge{red!50!black!70!}
\newcommand{\wh}[1]{\color{white}{#1}}

\begin{document}

\title[A Logic for Monitoring Dynamic Networks of Spatially-distributed CPS]{A Logic for Monitoring Dynamic Networks of Spatially-distributed Cyber-Physical Systems\rsuper*}
\titlecomment{{\lsuper*} This paper is an extended version of the "Monitoring mobile and spatially distributed cyber-physical systems'' manuscript published at MEMOCODE 2017~\cite{Bartocci17memocode}.}

\author[L.~Nenzi]{Laura Nenzi}[a,b]	
\address{University of Trieste, Italy}	
\email{lnenzi@units.it}  
\email{luca@dmi.units.it}

\author[E.~Bartocci]{Ezio Bartocci}[b]
\address{TU Wien, Austria}
\email{ezio.bartocci@tuwien.ac.at}

\author[L.~Bortolussi]{Luca Bortolussi}[a]

\author[M.~Loreti]{Michele Loreti}[c]	
\address{University of Camerino, Italy}	
\email{michele.loreti@unicam.it}  



 

\begin{abstract}
Cyber-Physical Systems (CPS) consist of inter-wined computational (cyber) and physical components interacting through sensors and/or actuators.  Computational elements are networked at every scale and can communicate with each other and with humans. Nodes can join and leave the network at any time or they can move to different spatial locations. In this scenario, monitoring spatial and temporal properties plays a key role in the understanding of how complex behaviors can emerge from local and dynamic interactions. 
We revisit here the Spatio-Temporal Reach and Escape Logic (STREL), a logic-based formal language designed to express and monitor spatio-temporal requirements over the execution of mobile and spatially distributed CPS.  STREL considers the physical space in which CPS entities (nodes of the graph) are arranged as a weighted graph representing their dynamic topological configuration. Both nodes and edges include attributes modeling physical and logical quantities that can evolve over time.
STREL combines the Signal Temporal Logic with two spatial modalities \emph{reach} and \emph{escape} that operate over the weighted graph. From these basic operators, we can derive other important spatial modalities such as \emph{everywhere}, \emph{somewhere} and \emph{surround}. We propose both qualitative and quantitative semantics based on constraint semiring algebraic structure.  We provide an offline monitoring algorithm for STREL and we show the feasibility of our approach with the application to two case studies: monitoring spatio-temporal requirements over a simulated mobile ad-hoc sensor network and a simulated epidemic spreading model for COVID19.
\end{abstract}



\maketitle
\section{Introduction}
From contact tracing devices preventing the epidemic spread to vehicular networks and smart cities, \emph{Cyber-Physical Systems} (CPS) are pervasive information and communication technologies augmenting the human perception and control over the physical world.  CPS consist of engineering, physical and biological systems that are tightly interacting with computational entities through sensors and actuators. CPS are networked at every scale and they are connected to the Internet (Internet of Things) enabling humans and other software systems to inter-operate with them through the World Wide Web.

A prominent example of CPS is present in the modern automotive systems where the substantial embedding of sensors, actuators, and computational units has facilitated the development of various driving assistance features such as the adaptive cruise control or the collision avoidance system. Furthermore, the upcoming 5G networks will empower soon vehicles also with the possibility to access real-time information about each other (position and speed of vehicles) and the condition of the other roads (accidents or traffic jams).
Thus, this dynamic network infrastructure promises to enhance further autonomous driving applications, reduce traffic congestion and improve safety.  

The benefits brought by this  increasingly pervasive technology  have also a price: unexpected failures can potentially manifest causing fatal accidents. Due to their safety-critical nature~\cite{RatasichKGGSB19}, engineers must verify that their behavior is correct with respect to rigorously defined requirements. However, given the uncertainty of the environment in which these systems operate, detecting all the failures at design time is usually unfeasible. Exhaustive verification techniques such as model checking are limited to very small instances due to state-space explosion.  An alternative  approach is to simulate a digital replica of the CPS (its digital twin) and to test the behavior under different scenarios. 
Requirements are generally  expressed in a formal specification language that can be monitored online (during the simulation) or offline over the simulation traces.
Most of the available specification languages~\cite{MalerN13,maler2016runtime,BartocciDGMNQ20,AsarinCM02}   can specify only temporal properties.  However, spatio-temporal patterns play a key role in the understanding of how emergent behaviors can arise from local interactions in such complex systems of systems.
Thus, an important problem is then how to specify in a formal language spatio-temporal requirements, and how to efficiently monitor them on the actual CPS or on the simulation of its  digital twin.

In this paper, we propose the Spatio-Temporal Reach and Escape Logic (STREL), a spatio-temporal specification language originally introduced in~\cite{Bartocci17memocode} and that we further revisit and extend in this manuscript.  STREL allows to specify spatio-temporal requirements and their monitoring over the execution of mobile and spatially distributed entities  and it is supported by the \textsc{MoonLight}~\cite{BartocciBLNS20} monitoring tool.

STREL considers the space, where the single components are spatially arranged, as a weighted graph representing the dynamical topological configuration.  Both nodes and edges include attributes modeling the physical and logical quantities that can evolve in time.
STREL combines the Signal Temporal Logic~\cite{Maler2004} with two spatial operators \emph{reach} and \emph{escape}. Other spatial operators such as \emph{everywhere}, \emph{somewhere} and \emph{surround} can be derived from them.
 The use of \emph{reach} and \emph{escape} opens the possibility to monitor the STREL property locally 
 at each node by observing the value of satisfaction of its neighbors.
The verdict of a STREL monitor can be 
 evaluated using different semantics (Boolean, real-valued)  based on the constraint semirings algebraic structures.

The use of STREL is by no means restricted to CPS as an application domain, but it is capable of capturing many interesting notions in other spatio-temporal systems, including biological systems (e.g. Turing patterns~\cite{BartocciBMNS15,Bartocci2016TNCS,NenziBCLM18}), epidemic spreading scenarios (in real space or on networks)~\cite{network_epidemic2015}, or ecological systems. In these cases, monitoring algorithms typically act as key subroutines in statistical model checking~\cite{BartocciBMNS15}.

This paper extends our preliminary work~\cite{Bartocci17memocode} as follows:
\begin{itemize}
\item we guide the reader through the all paper using a running example to facilitate the comprehension of our framework in each step;
\item we simplify the definition of dynamical spatial model and of the spatio-temporal semantics;
\item we extend spatial STREL operators to support interval constraints on distances;
\item we propose new monitoring algorithms, more efficient and able to work with interval constraints. We also provide correctness proofs and discuss in detail its algorithmic complexity; 
\item we consider a second case study where we monitor 
spatio-temporal requirements in STREL on a simulated epidemic spreading model for COVID19; 
\end{itemize}

The rest of the paper is organized as follows. We discuss 
the related work in Section~\ref{sec:related}. In Section~\ref{sec:runningExample} we introduce the reader with a running example while in Section~\ref{sec:definitions} we present the considered model of space and the type of signals. The STREL specification language is presented in Section~\ref{sec:ReachSTL} and the offline monitoring algorithms are discussed in Section ~\ref{sec:alg}. In Section~\ref{sec:ZigBee} and~\ref{sec:epidemic} we discuss the two case studies: the ZigBee protocol and the COVID-19 epidemic spreading. Finally, Section~\ref{sec:conclusion} presents the conclusions.

\section{Related Work}
\label{sec:related}

The study of spatial logics dates back to at least twenty years ago.
A first systematic study of spatial logics was proposed in a dedicated handbook~\cite{handbookSP}. These works discuss several important theoretical problems such as expressivity and decidability. However, the lack of verification procedures
and available tools made these specification
languages less practical to use.
More recently, Ciancia and others proposed spatial~\cite{ciancia2014} and spatial-temporal~\cite{CianciaGGLLM15}  model checking algorithms to verify the {\it Spatial Logic for Closure Spaces}~\cite{ciancia2014} (SLCS)
and its later extension~\cite{CianciaGGLLM15} with the temporal modality of the {\it  Computation Tree Logic}~\cite{EmersonH83}.
SLCS considers a discrete and topological notion of space that is based on closure spaces~\cite{Gal99}.
Another example of the spatial model checker is \textsc{VoxLogicA}~\cite{BelmonteCLM19,BuonamiciBCLM20}, a specialized tool designed for image analysis. \textsc{VoxLogicA} does not consider time and due to its specialization to image analysis, it cannot be employed 
to verify spatio-temporal properties
over cyber-physical systems.
Spatial-temporal model checking algorithms are also very computationally expensive due to the state-space explosion that is further exacerbated 
by the spatial domain. Here, we 
consider instead a more lightweight verification technique that consists in monitoring spatial properties over a single execution trace of the system. This approach 
can be applied both on a simulated model of the system or on the system
itself by properly instrumenting it and collecting its execution traces.

In the field of runtime verification  several logical frameworks for monitoring spatial-temporal properties~\cite{NenziBBLV20} are recently gaining momentum: {\it Spatial-Temporal Logic} (SpaTeL)~\cite{spatel}, {\it Signal Spatio-Temporal Logic} (SSTL)~\cite{NenziBCLM15}, {\it Spatial Aggregation Signal Temporal Logic}
(SaSTL)~\cite{MaBLS020,sastl2021} and {\it Spatio-Temporal  Reach and Escape Logic} (STREL)~\cite{Bartocci17memocode}. 
SpaTeL combines the 
{\it Signal Temporal Logic}~\cite{Maler2004} (STL) with the {\it Tree Spatial Superposition Logic}~\cite{bartocci2014} (TSSL).
Spatial patterns in TSSL are specified in terms of properties over quad trees~\cite{FinkelB74} spatial  data structures. 
Specifications in this logic can capture even complex fractal spatial patterns. However, finding the right formulation manually can be extremely cumbersome. 

SSTL~\cite{NenziBCLM15,NenziBCLM18} extends the STL temporal logic with two spatial modalities \emph{somewhere} and \emph{surrounded} that operate over a weighted graph where nodes represent spatial locations while edges their topological relations. The \emph{somewhere} operator specifies that a property is true in at least one of the nodes nearby, while the \emph{surround} operator specifies the property of being surrounded by a region of nodes satisfying its nested formula. SSTL can be interpreted using both  Boolean and real-valued semantics. One limitation of SSTL is that the topological relation of the nodes is fixed: the nodes cannot change their position.  STREL~\cite{Bartocci17memocode} overcomes this limitation by operating over a dynamic topological space. STREL introduces also two new spatial operators (\emph{reach} and \emph{escape}) that generalize and substitute the SSTL operators 
and it simplifies the monitoring procedure that is now computed locally in each node.    

SaSTL~\cite{MaBLS020,sastl2021} is yet another spatio-temporal monitoring language recently developed to monitor the internet of things of smart cities.  SaSTL is equipped with new logical 
spatial operators to express spatial aggregation and spatial counting and similarly to SSTL the monitoring procedure is limited to only a static topological space.


Another important key characteristic of STREL with respect to all the aforementioned spatio-temporal specification languages is the possibility to define the semantics using constraint semirings algebraic structures. This provides the possibility to elegantly define a unified monitoring framework for both the qualitative and quantitative semantics (similarly to~\cite{JaksicBGN18} for the STL case). 
Finally, it is also worth mentioning several recent applications of STREL ranging from mining spatio-temporal requirements from data~\cite{MohammadinejadD21} to synthesizing neural network-based controllers for multi-agent systems~\cite{strelnewsemantics}.

\section{Running Example: A Mobile ad hoc sensor network}
\label{sec:runningExample} 
A mobile ad-hoc sensor network (MANET) can consist of  up ten thousand mobile devices connected wirelessly, usually deployed to monitor environmental changes  such as pollution, humidity, light, and temperature.
Each sensor node is equipped with a sensing transducer, data processor, radio transceiver,  and an embedded battery.  A node can move independently in any direction and indeed can change its links to other devices frequently. 
Two nodes can communicate with each other if  their Euclidean distance is at most their  communication range as depicted in Figure~\ref{fig:proxconnect}~(right).
Moreover, the nodes can be of different type and their behavior and communication can depend on their types. In the next section, we consider the simplest MANET with all nodes of the same type, while in Section~\ref{sec:ReachSTL} we will differentiate them to describe more complex behaviors. 
\begin{figure}[ht]
    \centering
    \includegraphics[scale=0.5]{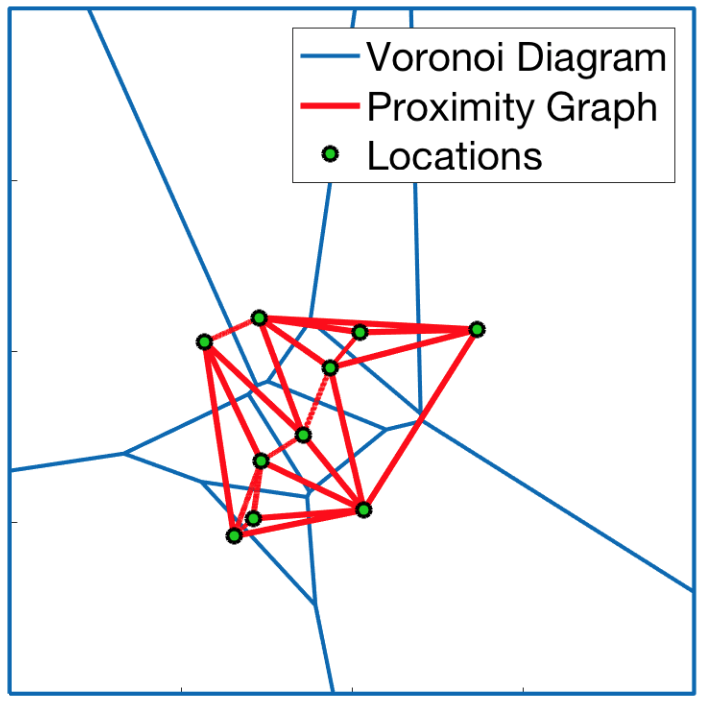}
    \includegraphics[scale=0.5]{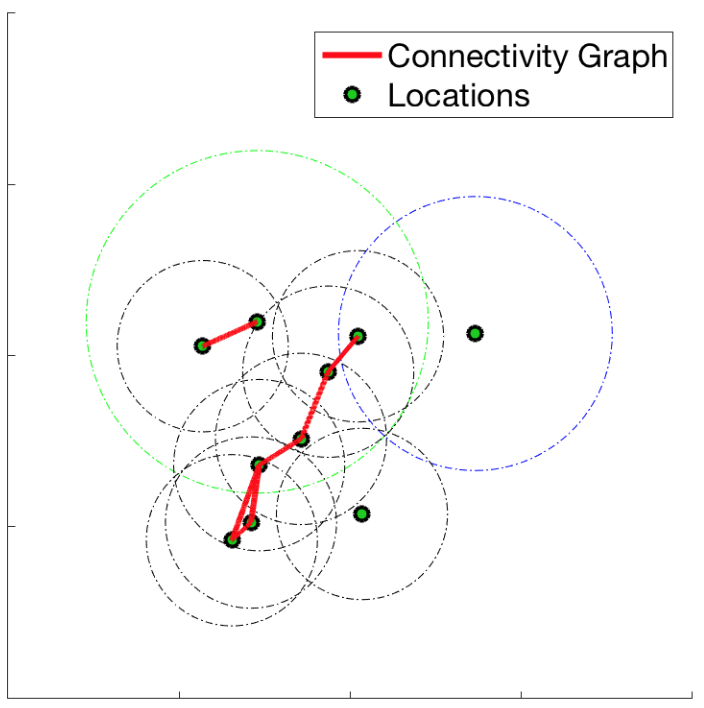}
    \caption{Proximity graph (left) and Connectivity graph (right)}
    \label{fig:proxconnect}
\end{figure}

\section{Spatial Models, Signals and Traces}
\label{sec:definitions}
In this section, we introduce the model of space we consider, 
and the type of signals that the logic specifies.

\subsection{Constraint Semirings}
An elegant and general way to represent the result of monitoring is based on \emph{constraint semiring}~\cite{BMR97}. This is an algebraic structure that consists 
of a domain and two operations named \emph{choose} and \emph{combine}. 
Constraint semirings are a subclass of semirings that have been shown 
to be very flexible, expressive, and convenient for a wide range of problems, 
in particular for optimization and solving problems with soft constraints 
and multiple criteria~\cite{BMR97}, and in model checking~{\protect{\cite{LM05}}}.

\begin{definition}[Semiring] 
A \emph{semiring} is a tuple $\langle A, \oplus, \otimes, \bot, \top \rangle$ composed by 
a set $A$, two operators $\oplus$, $\otimes$ and two constants $\bot$, 
$\top \in A$ such that: 
\begin{itemize}
\item $\oplus : A \times A \rightarrow A$ is an associative, commutative operator to ``choose'' among values\footnote{We let 
$x\oplus y$ to denote $\oplus(\{ x , y\})$.},
with $\bot$ as unit element ($\bot \oplus a = a,   \forall a \in A$).
\item $\otimes : A \times A \rightarrow A$ is an associative operator  to ``combine'' values with $\top$ as unit element ( $\top \otimes a = a,   \forall a \in A$) and $\bot$ as absorbing element ($\bot \otimes a = \bot,   \forall a \in A$ )
\item $\otimes$ distributes over $\oplus$;
\end{itemize}
\end{definition}

\begin{definition}[c-semiring]
A \emph{constraint semiring} (c-semiring)
is a semiring $\langle A, \oplus, \otimes, \bot, \top \rangle$ such that: 
\begin{itemize}
\item $\oplus$ is defined over $2^A$, idempotent ( $a\in A$ $a\oplus a=a\otimes a =a$) and has $\top$ as absorbing element ( $\top \oplus a = \top$)
\item $\otimes$  is commutative
\item $\sqsubseteq$, which is defined as $a\sqsubseteq b$ iff {$a\oplus b=b$}, 
provides a complete lattice $\langle A , \sqsubseteq , \bot, \top \rangle$. 
\end{itemize}
We say that a \emph{semiring}
$A$ is \emph{total} when $\sqsubseteq$ is a {\emph{total order}}. 
\end{definition}

With abuse of notation, we sometimes refer to a semiring 
$\langle A, \oplus,\otimes , \bot, \top  \rangle$ with the carrier $A$ 
and to its components by subscribing them with the carrier, i.e., 
$\oplus_A$, $\otimes_A$, $\bot_A$ and $\top_A$.  For the sake of a lighter notation, we drop the subscripts when clear from the context.

\begin{example}\label{ex:semirings}
Typical examples of c-semirings are\footnote{We use $\mathbb{R}^{\infty}$ (resp. $\mathbb{N}^{\infty}$) to denote $\mathbb{R}\cup\{-\infty,+\infty\}$ (resp. $\mathbb{N}\cup\{\infty\}$).}:
\begin{itemize}
%
\item the Boolean semiring $\langle  \{\mathit{true},\mathit{false}\}, \vee, \wedge, \mathit{false}, \mathit{true} \rangle$; 
\item the tropical semiring $\langle \mathbb{R}_{\geq 0}^{\infty},\emph{min},+,+\infty,0 \rangle$;
\item the max/min semiring: $\langle \mathbb{R}^{\infty}, \emph{max},\emph{min}, -\infty, +\infty \rangle$ ;
\item the integer semiring: $\langle \mathbb{N}^{\infty}, \emph{max},\emph{min}, 0, +\infty \rangle$.
%
%
\end{itemize}
All the above semirings are \emph{total}. 

\end{example}

One of the advantages of \emph{semirings} is that these can be easily composed. For instance, if $A$ and $B$ are two semirings, one can consider the \emph{cartesian product} $\langle A\times B,(\bot_A,\bot_B), (\top_A,\top_B), \oplus,\otimes\rangle$ where operations are applied elementwise.

\subsection{Spatial model}
Space is represented via a graph with edges having a weight from a set $A$.
We consider directed graphs (undirected graphs can consider symmetric relation).
\begin{definition}[$A-$spatial model] 
    An $A-$\emph{spatial model} $\mathcal{S}$ is a pair $\langle L, \wfun\rangle$ where:
    \begin{itemize}
        \item $L$ is a set of \emph{locations}, also named \emph{space universe};
        \item $\wfun\subseteq L\times A\times L$ is a \emph{proximity function} associating at most one label $w \in A$ with each distinct pair $\ell_1,\ell_2\in L$. 
    \end{itemize} 
\end{definition}

We will use $\mathbb{S}_{A}$ to denote the set of $A$-\emph{spatial models}, while $\mathbb{S}^{L}_{A}$ indicates the set of $A$-\emph{spatial models} having $L$ as a set of locations. In the following, we will equivalently write  $(\ell_1,w,\ell_2)\in \wfun$ as $\wfun(\ell_1,\ell_2)=w$ or  $\nextto{\ell_1}{w}{\ell_2}$, saying that $\ell_1$ is \emph{next to} $\ell_2$ with weight $w \in A$.


\begin{example}  $\mathbb{R}_{\geq 0}^{\infty}$-spatial model 
can be used to represent standard {\it weighed graphs} as Figure~\ref{fig:spmodel}. $L$ is the set of nodes and the proximity function $\wfun$ defines the weight of the edges, e.g. $\wfun(\ell_2,\ell_7)= \wfun(\ell_7,\ell_2) =5$.
{\small
 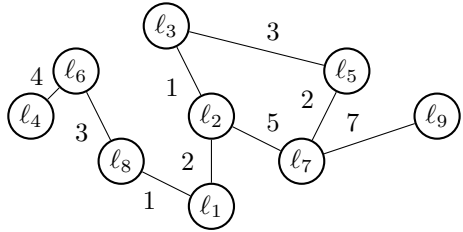
\begin{figure}[H]
\begin{center}
\begin{tikzpicture}
  [scale=.6,auto=left,every node/.style={circle,thick,inner sep=0pt,minimum size=6mm}]
  \node (1) [draw = black] at (-1,-1) {$\ell_1$};
  \node (2) [draw = black] at ( -1,1) {$\ell_2$};
  \node (3) [draw = black] at ( -2,3) {$\ell_3$};
  \node (4) [draw = black] at ( -5, 1){$\ell_4$};
  \node (5) [draw = black] at ( 2, 2) {$\ell_5$};
  \node (6) [draw = black] at (-4, 2) {$\ell_6$};
  \node (7) [draw = black] at (1, 0) {$\ell_7$};
  \node (8) [draw = black] at (-3,0) {$\ell_8$};
  \node (9)[draw = black] at (4,1) {$\ell_9$};

 \draw [-] (1) -- (2) node[midway] {2};
 \draw [-] (1) -- (8) node[midway] {1};
 \draw [-] (2) -- (7) node[midway] {5};
 \draw [-] (2) -- (3) node[midway] {1};
 \draw [-] (4) -- (6) node[midway] {4};
 \draw [-] (8) -- (6) node[midway] {3};
 \draw [-] (7) -- (9) node[midway] {7};
 \draw [-] (7) -- (5) node[midway] {2};
 \draw [-] (3) -- (5) node[midway] {3};

 \end{tikzpicture}
 \end{center}
     \caption{Example of a weighted undirected graph; e.g. $\wfun(\ell_2,\ell_7)= \wfun(\ell_7,\ell_2) =5$. }
 \label{fig:spmodel}
 \end{figure}
 }
\end{example}

A special class of spatial models is the ones based on \emph{Euclidean spaces}. 

\begin{definition}[Euclidean spatial model] 
\label{def:Euclidean}
Let $L$ be a set of locations, $R\subseteq L\times L$ a (reflexive) relation and $\lambda: L\rightarrow \mathbb{R}^{2}$ a function mapping each location to a point in  $\mathbb{R}^{2}$, we let $\mathcal{E}(L,R,\lambda)$ be the $\mathbb{R}^{\infty}\times\mathbb{R}^{\infty}$-spatial model\footnote{$\mathbb{R}^{\infty}$ is the \emph{max/min} semiring considered in Example~\ref{ex:semirings}.} $\langle L, \wfun^{\lambda, R}\rangle$ such that:
\begin{displaymath}
\wfun^{\lambda,R}=\{ (\ell_1,\lambda(\ell_1)-\lambda(\ell_2),\ell_2) | (\ell_1,\ell_2)\in R \}
\end{displaymath}
\label{def:euclisomod}
\end{definition}
Note that we label edges with a 2-dimensional vector $w$ describing how to reach $\ell_2$ from $\ell_1$, i.e.,  $\lambda(\ell_1) - w = \lambda(\ell_2)$. This obviously allows us to compute the Euclidean distance between  $\ell_1$ and $\ell_2$ as $\| w \|_2$, but, as we will see, allows us to compute the Euclidean distance of any pair of locations connected by any path, not necessarily by a line in the plane.
Note that, in the general case, a Euclidean spatial model is not a \emph{planar graph}.
A planar graph is obtained when, for instance, one considers the \emph{proximity graph} of Figure~\ref{fig:proxconnect}.

%
%
%
%
%
%
%
%
%

\begin{example}
\label{ex:manet} 
When considering a MANET, we can easily define different proximity functions for the same set of locations, where each location represents a mobile device.
Given a set of $n$ reference points in a two-dimensional Euclidean plane, a Voronoi diagram~\cite{Aurenhammer1991}  partitions the plane into a set of $n$ regions, one per reference point,  
assigning each point of the plane to the region corresponding to the closest reference point. 
The dual of the  Voronoi diagram is the proximity graph or Delaunay triangulation~\cite{Delaunay1934}.  
In Figure~\ref{fig:proxconnect} (left) we can see an example of the Voronoi diagram (in  blue) and proximity graph (in red). 
The proximity function can then be defined with respect to the Cartesian coordinates, as in Definition~\ref{def:euclisomod}: 
\begin{math}
\wfun^{\mu,R}(\ell_i,\ell_j)=\mu(\ell_i)-\mu(\ell_j)=(x_i,y_i)-(x_j,y_j)= (x_i - x_j , y_i -y_j)
\end{math}, where 
$(x_i,y_i)$ are the plane coordinates of the location $\ell_i$.

The proximity function can be also equal to a value that depends on other specific characteristics or behaviors of our nodes.  For instance, Figure~\ref{fig:proxconnect}~(right) represents the connectivity graph of MANET. In this case, a location $\ell_i$ is next to a location $\ell_j$ if and only if they are within their communication range.
\end{example}

Given an $A$-spatial model we can define \emph{routes}. 

\begin{definition}[Route]
Let $\mathcal{S}=\langle L,\wfun\rangle$, a \emph{route} $\route$ 
is an infinite sequence $\ell_0 \ell_1\cdots \ell_k\cdots$ in $L^{\omega}$ such that for any $i\geq 0$, $\nextto{\ell_i}{d}{\ell_{i+1}}$.
\end{definition}
Let $\route=\ell_0 \ell_1\cdots \ell_k\cdots$ be a route, $i\in \mathbb{N}$ and $\ell_i \in L$, we use:
\begin{itemize}
\item $\route[i]$ to denote the $i$-th node $\ell_i$ in $\route$;
\item $\route[i..]$ to indicate the suffix route $\ell_i \ell_{i+1} \cdots$;
\item $\ell \in \route$ when there exists an index $i$ such that $\route[i]=\ell$, while we use $\ell\not\in \route$ if this index does not exist;
\item $\route(\ell)$ to denote the first occurrence of $\ell$ in $\tau$:
\[
\route(\ell)=\left\{
\begin{array}{ll}
\min\{ i | \route[i]=\ell \} & \mbox{if $\ell\in \route$}\\
\infty & \mbox{otherwise} \\ 	
\end{array}
\right.
\]
\end{itemize}
We also use $Routes(\mathcal{S})$ to denote the set of routes in $\mathcal{S}$, while $Routes(\mathcal{S},\ell)$ denotes the set of routes starting from $\ell \in L$.

We can use routes to define the \emph{distance} among two locations in a \emph{spatial model}. This distance is computed via an appropriate function $f$ that combines all the weights in a route into a value taken from an appropriate \emph{total ordered} monoid $B$. 

\begin{definition}[Distance Domain]
\label{def:distDom}
We say that  \emph{distance domain} $(B,\bot_B,\top_B,+_{B},\leq_{B})$ whenever $\leq_{B}$ is a total order relation over $B$ where $\bot_{B}$ is the minimum while $\top_{B}$ is the maximum and $(B,\bot_B,+_{B})$ is a monoid. Given a distance domain $B$, we use $\bot_{B}$, $\top_{B}$, $+_B$ and $\leq_B$ to denote its elements.
\end{definition}

\begin{definition}[Distance Function and Distance over paths]
    Let  $\mathcal{S}=\langle L,\wfun\rangle$ be an $A$-spatial model, $\route$ a route in $\mathcal{S}$, $\langle B,\bot_{B},\top_{B},+_{B},\leq_{B} \rangle$ a \emph{distance domain},  we call $f:A\rightarrow B$ the \emph{distance function}, associating elements of $A$ to the distance domain $B$. The distance $d_{\route}^{f}[i]$ up-to index $i\in \mathbb{N}^{\infty}$ is defined as follows:
	\[
	d_{\route}^{f}[i]= \begin{cases}
      \bot_{B} &   i=0 \\
      \infty & i=\infty \\
      f(w) +_{B} d_{\route[1..]}^{f}[i-1] & (i>0)\mbox{ and } \nextto{\route[0]}{w}{\route[1]} 
\end{cases} \\	
	\]
\noindent	
Given a locations $\ell\in L$, the distance over $\route$ up-to $\ell$ is then $d_{\route}^{f}[\ell]  = d_{\route}^{f}[\route(\ell)]$ if $\ell\in \route$, while it is $\top_{B}$ if $\ell\not\in \route$.
%
\end{definition}


\begin{example}
\label{ex:distancefunction}
Considering again  a MANET, one could be interested in different types of distances, e.g., 
\emph{counting} the number of \emph{hop},  or distances induced by the weights of the Euclidean space structure.
%

\noindent
To count the number of hop, we can simply use the  function 
$hop: A \rightarrow \mathbb{N}^{\infty}$, taking values in the distance domain $\langle \mathbb{N}^{\infty}, 0, \infty, +, \leq \rangle$:
\[
hop(w)=1
\]   
and in this case $d^{hop}_\tau[i]=i$.

Considering the proximity function $\wfun^{\mu,R}(\ell_i,\ell_j)$ computed from the Cartesian coordinates and  the distance domain $\langle \mathbb{R}^{\infty}, 0, \infty, +, \leq \rangle$, we can use the Euclidean distance $\Delta(x,y)=  \| (x,y) \|$, where $(x,y)$ are the  coordinates of the vectors returned by $\wfun^{\mu,R}$. 

It is easy to see that for any route $\route$  and for any  location $\ell \in L$ in $\route$, the function $d_{\route}^{\Delta}(\ell)$ yields the sum of lengths of the edges in $\mathbb{R}^{2}$ connecting $\ell $ to $\route(0)$.

 
%
%

Given a distance function $f: A\rightarrow B$, the distance between two locations $\ell_1$ and $\ell_2$ in a $A$-spatial model is obtained by choosing the minimum distance along all possible routes starting from $\ell_1$ and ending in $\ell_2$:
\[d^{f}_{\mathcal{S}}[\ell_1,\ell_2] = \min\left\{ d^{f}_{\route}[\ell_2] | \route\in Routes(\mathcal{S},\ell_1) \right\}.
\]

\begin{example}
\label{ex:distanceLocations}
Consider again the distance functions defined for a MANETS. For \emph{hop}, we are taking the minimum hop-length over all paths connecting $\ell_1$ and $\ell_2$, resulting in the shortest path distance. 
\end{example}

	
%
%
%
%
\end{example}

\subsection{Spatio-Temporal Signals}

\begin{definition}
A 
{\emph{signal domain}} is a tuple $\langle D, \oplus,\otimes, \odot,\bot,\top\rangle$ where:
\begin{itemize}
	\item $\langle D, \oplus,\otimes,\bot, \top\rangle$, is a \emph{c-semiring};
	\item $\odot: D\rightarrow D$, is a \emph{negation function} such that:
	\begin{itemize}
	\item $\odot\top =\bot$;
	\item $\odot\bot = \top$;
	\item $\odot(v_1\oplus v_2)=(\odot v_1)\otimes (\odot v_2)$
	\item $\odot(v_1\otimes v_2)=(\odot v_1)\oplus (\odot v_2)$
	\item for any $v\in D$, $\odot ( \odot v ) = v$.
	\end{itemize}	
\end{itemize}	
\end{definition}

In this paper, we consider two \emph{signal domains}:
\begin{itemize}
	\item Boolean signal domain $\langle \{ \top , \bot \}, \vee, \wedge,\neg,  ,\bot,\top, \rangle$ for qualitative monitoring;
	\item {Max/min signal domain $\langle \mathbb{R}^{\infty}, \max, \min, -, \bot, \top,\rangle$} for quantitative monitoring.
\end{itemize}

For signal domains, we use the same notation and notational conventions introduced for semirings. 

\begin{definition} Let $\mathbb{T}=[0, T] \subseteq  \mathbb{R}_{\geq 0}$ a time domain and $\langle D, \oplus,\otimes, \odot ,\bot ,\top \rangle$ a \emph{signal domain}, a \emph{temporal $D$-signal} $\tsign$ is a function
$\tsign: \mathbb{T}\rightarrow D$.

\noindent
Consider a finite sequence: 
\[
\pct = [(t_{0}, d_0),\ldots,(t_{n}, d_{n})]
\]
such that for $\forall i, t_i\in \mathbb{T}$, $t_i<t_{i+1}$ and $d_i\in D$.
We let $\pct$ denote a \emph{piecewise constant temporal $D$-signal} in $\mathbb{T}=[t_0, T]$, that is 
  \[
  \pct(t) = \begin{cases}
      & d_i  \quad \text{ for } t_{i} \leq t < t_{i+1}, \\
      & d_n \quad \text{ for } t_{n}  \leq t \leq T;
\end{cases} \\
\]
\end{definition}

Given a \emph{piecewise constant temporal signal} $\pct=[(t_{0}, d_0),\ldots,(t_{n}, d_{n})]$ we use $\mathcal{T}( \pct )$ to denote the set $\{ t_0,\ldots, t_n \}$ of \emph{time steps} in $\pct$; $start(\pct)$ to denote $t_0$; while we say that $\pct$ is \emph{minimal} if and only if for any $i$, $d_i\not=d_{i+1}$. 
We also let  $\pct[ t=d ]$ denote the signal obtained from $\pct$ by adding the element $(t,d)$. 
Finally, if $\tsign_1$ and $\tsign_2$ are two $D$-temporal signals, and $op: D\times D\rightarrow D$, $\tsign_1~op~\tsign_2$ denotes the signal associating with each time $t$ the value $\tsign_1(t)~op~\tsign_2(t)$. Similarly, if $op:D_1 \rightarrow D_2$, $op~\tsign_1$ denotes the $D_2-$signal associating with $t$ the value $op~ \tsign_1(t)$.

\begin{definition}[Spatial $D$-signal] Let $L$ be a \emph{space universe}, and $\langle D, \oplus,\otimes, \odot ,\bot ,\top\rangle$ a signal domain. A \emph{spatial $D$-signal} is a function $\ssign: L\rightarrow D$.  
\end{definition}

\begin{definition}[Spatio-temporal $D$-signal]
  Let $L$ be a space universe, $\mathbb{T}=[0, T]$ a time domain, and $\langle D, \oplus,\otimes, \odot,\top,\bot\rangle$ a signal domain,   a spatio-temporal $D$-signal is a function
  \[ \sts: L \rightarrow \mathbb{T} \rightarrow D \]
\noindent
such that $\sts(\ell)=\tsign$ is a temporal signal that returns a value $\tsign(t) \in {D}$ for each time $t  \in \mathbb{T}$. We say that $\sts$ is \emph{piecewise constant} when for any $\ell$, $\sts(\ell)$ is a \emph{piecewise constant temporal signal}. \emph{Piecewise constants} signal are denoted by $\pcsts$. Moreover, we use $\mathcal{T}(\pcsts)$ to denote $\bigcup_{\ell}\mathcal{T}(\pcsts(\ell))$. Finally, we let $\pcsts@t$ denote the spatial signal associating each location $\ell$ with $\pcsts(\ell,t)$.
\end{definition}

Given a spatio-temporal signal $\sts$, we use $\sts@t$
to denote the \emph{spatial signal} at time $t$, i.e. the signal $\ssign$ such that $\ssign(\ell)=\sts(\ell)(t)$, for any $\ell \in L$.
Different kinds of signals can be considered while the signal domain $D$ is changed. Signals with  $D= \{ true , false \}$  are called \emph{Boolean signals}; with $D = \mathbb{R}^{\infty}$ are called real-valued or \emph{quantitative signals}.

 \begin{definition}[$D$-Trace]
Let $L$ be a space universe, a {\it spatio-temporal $D$-trace} is a function
$\vec x: L \rightarrow \mathbb{T} \rightarrow D_1 \times \cdots \times D_n$ 
such that for any $\ell\in L$ yields a vector of  temporal signals $\vec{x}(\ell)=(\tsign_1,\ldots,\tsign_n)$.
%
In the rest of the paper, we use $\vec{x}(\ell,t)$ to denote $\vec{x}(\ell)(t)$.
\end{definition}

\begin{example}
 We can consider a $(\mathbb{R} \times \mathbb{R})$-spatio-temporal trace of our sensor network as $\vec x: L \rightarrow \mathbb{T} \rightarrow \mathbb{R} \times \mathbb{R}$ that associates a set of temporal signals $\vec x(\ell)= (\tsign_B,\tsign_T)$ at each location, where $\tsign_B$ and  $\tsign_T$ respectively  correspond to the temporal signals of the battery and the temperature in location $\ell$, and each signal has domain $\langle \mathbb{R}, \max, \min, -, \bot, \top,\rangle$.

\end{example}

We plan to work with spatial models that can dynamically change their configurations. For this reason, we need to define a function that returns the spatial configuration at each time.

\begin{definition}[Dynamical $A$-Spatial Model]
Let $L$ be a spatial universe, a {\it dynamical $A$-spatial model} is a function $\lserv : \mathbb{T}\rightarrow  \mathbb{S}^{L}_{A}$ associating  each element in the time domain $\mathbb{T}$ with $A$-spatial model  $\lserv(t)=\langle L, \wfun\rangle$ that describes the spatial configuration of locations. 
\end{definition}

With an abuse of notation, we use $\mathcal{S}$ for both a dynamical spatial model and a static spatial model, where, for any $t$, $\mathcal{S} =\lserv(t)$. 

\begin{example}
Let us consider a MANET with a proximity graph.
Figure~\ref{fig:voronoimobility} shows two different snapshots, $\mathcal{S}(t_1)=\langle L,\wfun_1 \rangle$ and $\mathcal{S}(t_2)=\langle L,\wfun_2 \rangle$, of the the dynamical spatial model $S$ for time $t_1$ and $t_2$. We can see that locations $\ell_1$ and $\ell_2$ change their position, this changed also the Voronoi diagram and the proximity graph. 
\begin{figure}[H]
    \centering
    \includegraphics[scale=0.4]{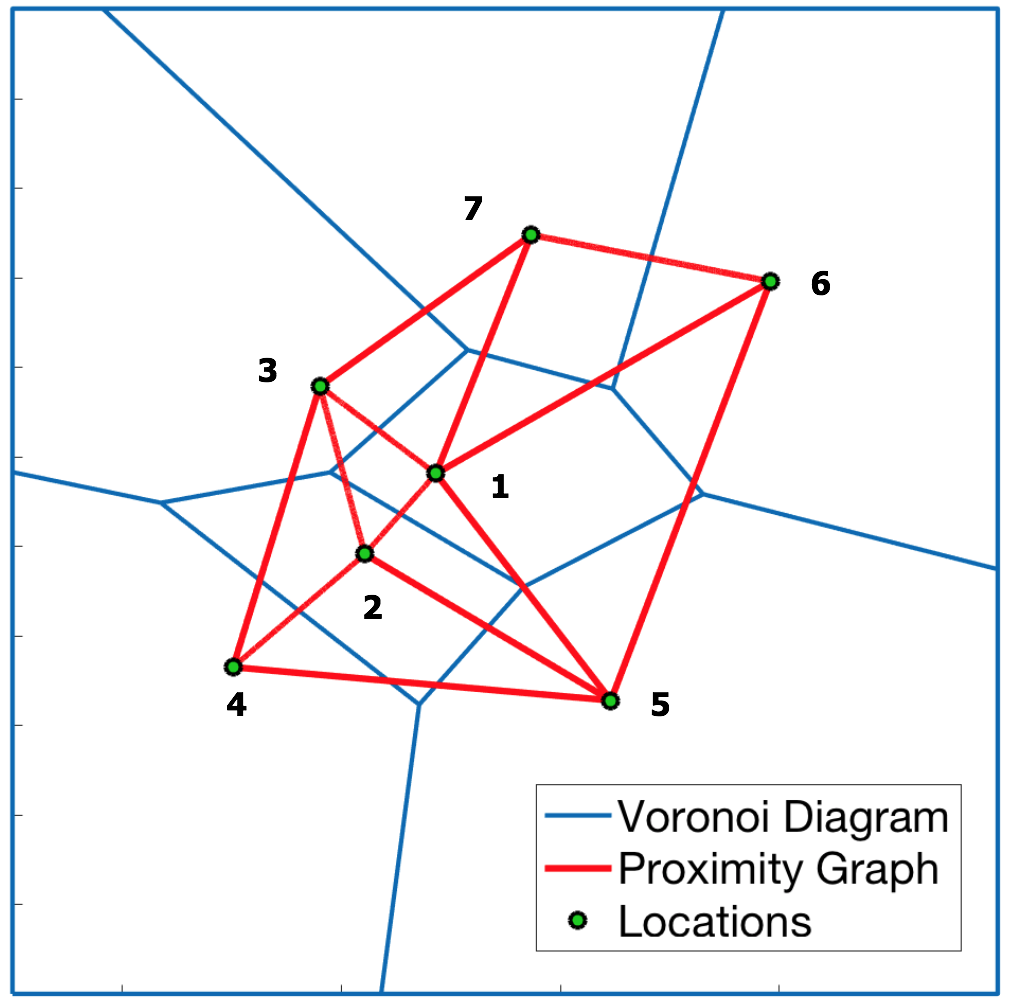}
    \includegraphics[scale=0.39]{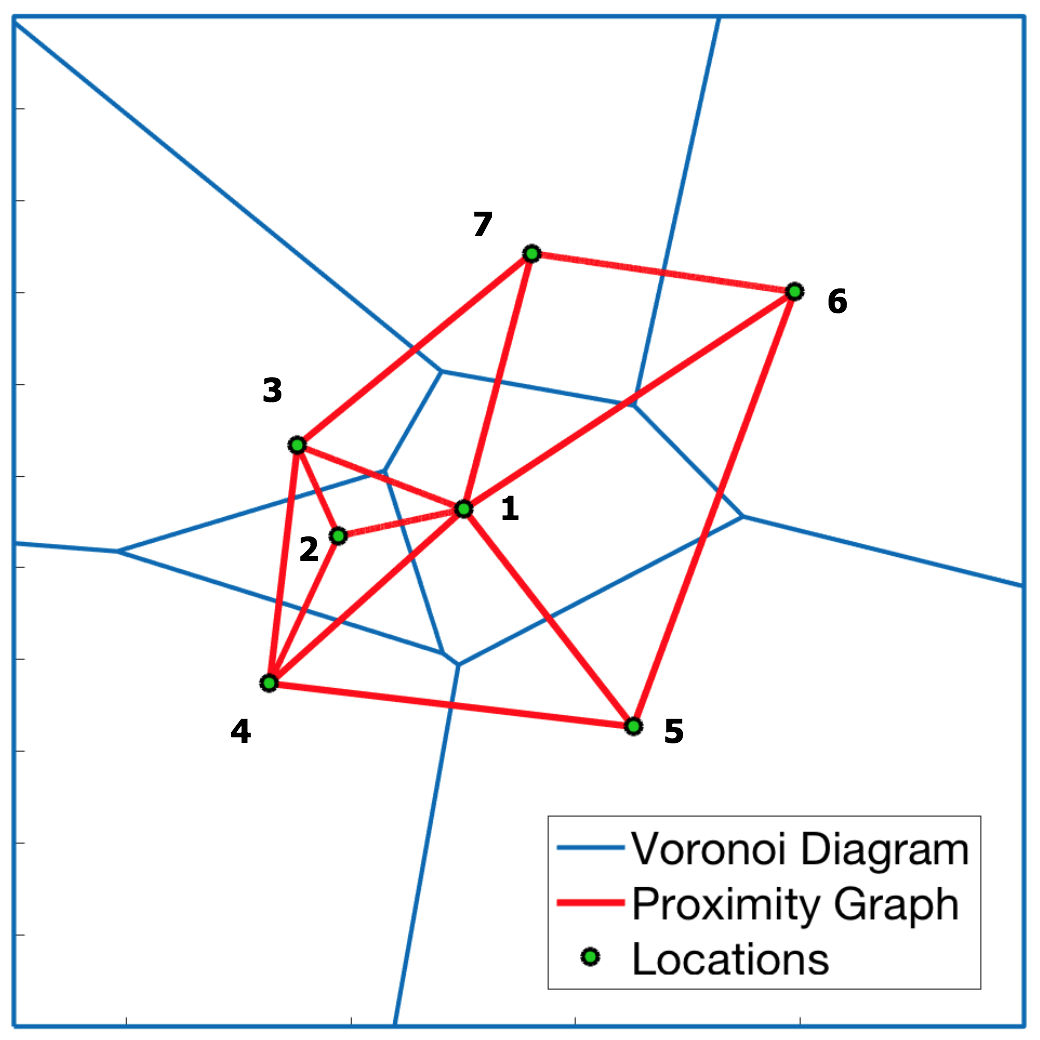}
    \caption{Two snapshots of a spatial model with 7 locations $\ell_1,\dots,\ell_7$ that move in a 2D Euclidean space. The plane is partitioned using a Voronoi Diagram ({\color{blue} blue}). In {\color{red} red} we have the proximity graph.}
    \label{fig:voronoimobility}
\end{figure}
\end{example}

\section{Spatio-temporal Reach and Escape Logic}
\label{sec:ReachSTL}
In this section,  we present the {\it Spatio-Temporal Reach 
and Escape Logic}~(STREL), an extension of the {\it Signal Temporal Logic}. 
We define the syntax and the semantics of STREL, describing 
in detail the spatial operators and their expressiveness. 

The syntax of STREL is given by
{\small
\[
\varphi :=  
    \mu \mid  
    \neg \varphi \mid  
    \varphi_{1} \wedge \varphi_{2} \mid  
    \varphi_{1} \: \until{[t_{1},t_{2}]} \: \varphi_{2} \mid 
    \varphi_{1} \: \since{[t_{1},t_{2}]} \: \varphi_{2} \mid  
    \varphi_{1} \: \reach{[d_{1},d_{2}]}{f:A\rightarrow B} \: \varphi_{2} \mid 
    \escape{[d_{1},d_{2}]}{f:A\rightarrow B}  \: \varphi
\]
}
where  $\mu$ is an {\it atomic predicate} ($AP$), {\it negation} $\neg$ and {\it conjunction} $\wedge$ are the standard Boolean connectives, $\until{[t_{1},t_{2}]}$ and  $\since{[t_{1},t_{2}]}$ are the {\it Until} and the {\it Since} temporal modalities, with $[t_{1},t_{2}]$ a real positive closed interval. For more details about the temporal operators, we refer  the reader to~\cite{MalerN13, Maler2004, Donze2013}.
The spatial modalities are the  {\it reachability} $\reach{[d_{1},d_{2}]}{f:A\rightarrow B}$ and the {\it escape} $\escape{[d_{1},d_{2}]}{f:A\rightarrow B}$ operators, with $f:A \rightarrow B$ a \emph{distance function} (see Definition~\ref{ex:distancefunction}), $B$ a \emph{distance domain}, and $d_{1},d_{2}\in B$ with $d_1\leq_{B} d_2$.
In what follows, we omit the type info about function $f$ when it is clear from the context or where it does not play any role.

The reachability operator $\phi_1 \reach{[d_1, d_2]}{f}\phi_2$ describes the behavior of reaching a location satisfying property $\phi_2$, through a path with all locations that satisfy  $\phi_1$, and with a  distance that belongs to $[d_1, d_2]$.
The escape operator $\escape{[d_1, d_2]}f{\phi}$, instead, describes the possibility of escaping from a certain region via a route passing only through locations that satisfy $\phi$,  with the distance between the starting location of the path and the last that belongs to the interval $[d_1, d_2]$. Note that the main difference between these two operators is that the distance of the reach operator is with respect to the path, instead, the distance of the escape operator is between the locations, so it considers the shortest path distance between the starting location and the last.


As customary, we can derive the {\it disjunction} operator $\vee$ and the future  {\it eventually} $\ev{[t_{1},t_{2}]}$ and  {\it always} $\glob{[t_{1},t_{2}]}$ operators from  the until temporal modality, and the corresponding past variants from the since temporal modality, see~\cite{MalerN13} for details. 
We can define also other three derived spatial operators: the {\it somewhere} $\somewhere{\leq d}{f: A\rightarrow B}\phi$ and the {\it everywhere} $\everywhere{\leq d}{f: A\rightarrow B}\phi$ that describe behaviors of some or of all locations at a certain distance from a specific point, and the {\it surround} that expresses the topological notion of being surrounded by a $\phi_2$-region, while being in a $\phi_1$-region, with additional metric constraints. A more thorough discussion of the spatial operators will be given after introducing the semantics.

\subsection{Semantics}
The semantics of STREL is  evaluated point-wise  at each time and at each location. We stress that each STREL formula $\varphi$ abstracts from the specific domain used to express the satisfaction value of $\varphi$.
These, of course, are needed to define the semantics.  In the following, we assume that $D_1$ is the domain of the spatio-temporal traces, $D_2$ is the semiring where the logic is evaluated and $B$ is a distance domain as defined in Definition~\ref{def:distDom}. 

\begin{definition} [Semantics] 
\label{generalsemantics}
Let $\mathcal{S}$ be a dynamical $A$-spatial model with $L$ the space universe,  $D_1$ and $D_2$ be two signal domains, and $\vec x$ be a {\it spatio-temporal $D_1$-trace} for $L$.
The $D_2$-monitoring function $\fmon$ of $\vec x$ is recursively defined in Table~\ref{tab:monitoring}.

\end{definition}

\begin{table}[ht]
{\footnotesize
\begin{minipage}[b]{1.\linewidth}
  \hspace{-0.0cm}
\begin{tabular}[b]{rcl}

$\fmon( \lserv, \vec{x}, \mu, t, \ell)$ & $=$ & $\iota(\mu,\vec{x}(\ell,t))$ \\[.2cm]

$\fmon(\lserv, \vec{x}, \neg\varphi, t, \ell)$ & $=$ & $\odot_{D_{2}} \fmon(\lserv, \vec{x}, \varphi, t, \ell)$  \\[.2cm]

$\fmon( \lserv, \vec{x}, \varphi_1 \wedge \varphi_2, t, \ell)$ & $=$ & $\fmon( \lserv, \vec{x}, \varphi_1, t, \ell) \otimes_{D_2} \fmon(\lserv, \vec{x}, \varphi_2, t, \ell)$  \\[.2cm]  

$\fmon( \lserv, \vec{x}, \varphi_{1} \: \until{[t_{1},t_{2}]} \: \varphi_{2}, t, \ell)$ & $=$ & ${\bigoplus_{D_2}}_{t' \in [t+t_{1}, t+t_{2}]} \big (\fmon(  \lserv, \vec{x}, \varphi_2, t', \ell) \otimes_{D_2} {\bigotimes_{D_2}}_{t'' \in [t, t']} \fmon( \lserv, \vec{x}, \varphi_1, t'', \ell) \big)   $ \\[.2cm]
      
$\fmon( \lserv, \vec{x}, \varphi_{1} \: \since{[t_{1},t_{2}]} \: \varphi_{2}, t, \ell)$ & $=$ &  
${\bigoplus_{D_2}}_{t' \in  [t-t_{2}, t-t_{1}]} \big (\fmon(\lserv, \vec{x}, \varphi_2, t', \ell) \otimes_{D_2} {\bigotimes_{D_2}}_{t'' \in [t', t]} \fmon(  \lserv, \vec{x}, \varphi_1, t'', \ell) \big)   $ \\[.2cm]

 $\fmon(\lserv, \vec{x}, \varphi_{1} \: \reach{[d_{1},d_{2}]}{f:A\rightarrow B} \: \varphi_{2}, t, \ell)$ & $=$ \\
\multicolumn{3}{r}{
${\bigoplus_{D_2}}_{\tau\in Routes(\lserv(t),\ell)}
        ~~{ \bigoplus_{D_2}}_{i : \left(d_{\tau}^{f}[i] \in [d_{1},d_{2}]\right)}
        \left(
            \fmon(  \lserv, \vec{x}, \varphi_2, t, \tau[i])
            \otimes_{D_{2}}
            {\bigotimes_{D_2}}_{j < i} 
                \fmon(  \lserv, \vec{x}, \varphi_1, t, \tau[j])         %
        \right)$} \\[.2cm]
$\fmon( \lserv, \vec{x}, \escape{[d_{1},d_{2}]}{f:A\rightarrow B} \: \varphi, t, \ell)$ & $=$ &
${\bigoplus_{D_2}}_{\tau\in Routes(\lserv(t),\ell)}
        ~~{\bigoplus_{D_2}}_{\ell'\in \tau : \left(d_{\lserv(t)}^{f}[\ell,\ell'] \in [d_{1},d_{2}]\right)}
            ~~{\bigotimes_{D_2}}_{i \leq \tau(\ell')} 
                \fmon(  \lserv, \vec{x}, \varphi, t, \tau[i])$ \\[.2cm]
\end{tabular}
\end{minipage}
	\caption{Monitoring function.}
	\label{tab:monitoring}
	}
\end{table}



Given a formula $\phi$, the function $\fmon( \lserv, \vec{x}, \phi, t, \ell)$ corresponds to the evaluation of the formula at time $t$ in the location $\ell$. 
The procedure is exactly the same for different choices of the formula evaluation domain, just operators have to be interpreted according to the chosen semirings and signal domains. 
In particular, the choice of the signal domain $D_2$ produces different types of semantics. 
In this paper, we consider two signal domains: $\mathbb{B}$ and $\mathbb{R^\infty}$, giving rise to qualitative and quantitative monitoring, corresponding respectively to a Boolean answer and a real satisfaction value. 
For the
Boolean signal domain ($D_2 = \langle \{ \top , \bot \}, \vee, \wedge,\neg \rangle $ ), 
 we say that $(\lserv , \vec{x}(\ell,t))$ satisfies a formula $\phi$ iff $\fmon( \lserv, \vec{x}, \phi, t, \ell)= \top$. 
 For  {max/min signal domain $\langle \mathbb{R}^{\infty}, \max, \min, -, \bot, \top,\rangle$} we say that $(\lserv , \vec{x}(\ell,t))$ satisfies a formula $\phi$ iff
 $\fmon( \lserv, \vec{x}, \phi, t, \ell) > 0$.
In the following, we use $\tilde{\sigma}^{\lserv,\vec{x}}_{\phi}$ to denote the spatio-temporal $D_2$-signal such that for any $t$ and $\ell$ $\fmon( \lserv, \vec{x}, \phi, t, \ell)=\tilde{\sigma}^{\lserv,\vec{x}}_{\phi}(\ell, t)$.

We describe now in detail the semantics through the sensor network example as the system on which we specify our properties, in particular, we use  the graph in Figure~\ref{fig:spprop} to describe the spatial operators.
\begin{example}[ZigBee protocol]
\label{ex:zigbee}
 In Figure~\ref{fig:spprop}, the graph represents a MANET. In particular, we consider the nodes with three different  roles such as the ones implemented in the ZigBee protocol: {\it coordinator}, {\it router} and {\it EndDevice}. The Coordinator node $(\coord)$, represented in green color in the graph, is unique in each network and is responsible to initialize the network.  After the initialization of the network has started, the coordinator behaves as a router. 
The Router node $(\router)$, represented in red color in the graph, acts as an intermediate router, passing on data from other devices. The EndDevice node $(\device)$, represented in blue, can communicate 
only with a parent node (either the Coordinator or a Router) and it is unable to relay data from other devices.
    Nodes move in  space and the figure corresponds to the spatial configuration at a fixed time $t$. 
    As spatio-temporal trace, let us consider a  $\{ \coord, \router, \device \}  \times \mathbb{R}$-trace $\vec x: L \rightarrow \mathbb{T} \rightarrow \mathbb{Z} \times \mathbb{R}^{\infty}$ denoting the pair: (kind of node, level of battery), i.e. $\vec x (\ell, t)= (\coord, b)$ if $\ell$ is a coordinator, $\vec x (\ell, t)= ( \router,b) $ if $\ell$ is a router, and  $\vec x (\ell, t)= (\device,b) $ if $\ell$ is an end node, where $b$ is the level of the battery. 
\end{example}

\noindent{\bf Atomic Proposition.} 
The function $\iota: AP\times D_1^{n} \rightarrow D_2$  is the \emph{signal interpretation function} and permits to translate the input trace in a different  ${D}_2$-spatio temporal signal for each atomic proposition in $AP$, which will be the input of the monitoring procedure.
Different types of atomic propositions and signal interpretations are admissible.  E.g., we can simply consider a finite set 
$\{p_1, \dots, p_n \}=AP$ and an interpretation function $\iota(p_i,\vec x(\ell,t))=\top$ iff $x_i(\ell,t)=\top$. In Figure~\ref{fig:spprop}, we can consider atomic propositions describing the type of node, i.e., the Boolean propositions $\{ \coord, \router, \device \}$ are true if the node is of the corresponding type. In case of real valued signals and of a quantitative interpretation of the logic ($D_2$ being in this case the real valued max/min semiring), we  can consider inequalities $\mu=(g(\vec{x})\geq 0)$ for some real function $g$ and define $\iota(\mu,\vec x(\ell,t))=g(\vec x(\ell,t))$, e.g. $b > 0.5$, that means "the level of the battery is greater than $50\%$

\noindent{\bf Negation.} The negation operator is interpreted with the negation function $\odot_{D_{2}}$ of the chosen signal domain; e.g. 
$\fmon(\lserv, \vec{x}, \neg\varphi, t, \ell)= \neg \fmon( \lserv, \vec{x}, \varphi, t, \ell)$ for the Boolean signal domain and $\fmon(\lserv, \vec{x}, \neg\varphi, t, \ell)= - \fmon( \lserv, \vec{x}, \varphi, t, \ell)$ for the quantitative ones.

\noindent{\bf Conjunction and Disjunction} The conjunction operator $\varphi_1 \wedge \varphi_2$ is interpreted with the combine operator $\otimes_{D_2}$,  i.e. $\fmon( \lserv, \vec{x}, \varphi_1 \wedge \varphi_2, t, \ell) = \fmon( \lserv, \vec{x}, \varphi_1, t, \ell)  \otimes_{D_2} \fmon(\lserv, \vec{x}, \varphi_2, t, \ell)$, which corresponds to wedge $\wedge$ operator for the Boolean semantics. This means that $\fmon( \lserv, \vec{x}, \varphi_1 \wedge \varphi_2, t, \ell)=1$ iff both $\fmon( \lserv, \vec{x}, \varphi_1)$ and 
$\fmon( \lserv, \vec{x}, \varphi_2)$ are equal to 1. 
For the quantitative semantics $\otimes_{D_2}$ is interpreted ad the minimum operator, so, $\fmon( \lserv, \vec{x}, \varphi_1 \wedge \varphi_2, t, \ell) = \min(\fmon( \lserv, \vec{x}, \varphi_1, t, \ell) , \fmon(\lserv, \vec{x}, \varphi_2, t, \ell))$. Similarly the disjunction $\varphi_1 \vee \varphi_2$ is interpreted through the choose operator $\oplus_{D_2}$,  i.e. $\fmon( \lserv, \vec{x}, \varphi_1 \vee \varphi_2, t, \ell) = \fmon( \lserv, \vec{x}, \varphi_1, t, \ell)  \oplus_{D_2} \fmon(\lserv, \vec{x}, \varphi_2, t, \ell)$, which corresponds to the $\vee$ for the Boolean semantics and to the $\max$ for the quantitative one. 

In the rest of the description, we focus on the Boolean semantics, i.e. $D_2 = \langle \{ \top , \bot \}, \vee, \wedge,\neg \rangle $, the Quantitative semantics can be derived substituting $\vee, \wedge$ with $\min, \max$, as seen for conjunction and disjunction.

\noindent{\bf Until.} \[\fmon( \lserv, \vec{x}, \varphi_{1} \: \until{[t_{1},t_{2}]} \: \varphi_{2}, t, \ell) = \bigvee_{t' \in t + [t_{1}, t_{2}]} \linebreak (\fmon(  \lserv, \vec{x}, \varphi_2, t', \ell) \wedge \bigwedge_{t'' \in [t, t']} \fmon( \lserv, \vec{x}, \varphi_1, t'', \ell) \big) .\]
As customary, $(\lserv, \vec{x}(\ell,t))$ satisfies $\varphi_{1} \: \until{[t_{1},t_{2}]} \: \varphi_{2}$ iff it satisfies $\varphi_{1}$ from $t$ until, in a time between $t_{1}$ and $t_{2}$ time units in the future, $\varphi_{2}$ becomes true. Note how the temporal operators are evaluated in each location separately.

\noindent{\bf Since.} \[\fmon( \lserv, \vec{x}, \varphi_{1} \: \since{[t_{1},t_{2}]} \: \varphi_{2}, t, \ell)  = \bigvee_{t' \in t - [-t_{2}, -t_{1}]} \linebreak \big (\fmon(\lserv, \vec{x}, \varphi_2, t', \ell) \wedge \bigwedge_{t'' \in [t', t]} \fmon(  \lserv, \vec{x}, \varphi_1, t'', \ell) \big)\]
 $(\lserv, \vec{x}(\ell,t))$ satisfies $\varphi_{1} \: \since{[t_{1},t_{2}]} \: \varphi_{2}$ iff it satisfies $\varphi_{1}$ from now since, in a time between $t_{1}$ and $t_{2}$ time units in the past, $\varphi_{2}$ was true.

Except for the interpretation function, the semantics of the Boolean and the temporal operators is directly derived from and coincident with that of STL (qualitative for Boolean signal domain and quantitative for an $\mathbb{R}^\infty$ signal domain), see~\cite{Maler2004, Donze2013} for details. 


{\small
\noindent{\bf Reachability.} \[\fmon(\lserv, \vec{x}, \varphi_{1} \: \reach{[d_{1},d_{2}]}{f} \: \varphi_{2}, t, \ell)=
\bigvee_{\tau\in Routes(\lserv(t),\ell)} \linebreak
        \bigvee_{i:\left(d_{\tau}^{f}[i] 
        \in [d_{1},d_{2}]\right)} 
        ( \fmon(  \lserv, \vec{x}, \varphi_2, t, \tau[i])
            \wedge 
            \bigwedge_{j < i}  \fmon(  \lserv, \vec{x}, \varphi_1, t, \tau[j]) )\]
}

\noindent $(\lserv , \vec{x}(\ell,t))$ satisfies
$\varphi_{1} \: \reach{d}{f} \: \varphi_{2}$ iff it satisfies $\varphi_2$ in a location $\ell'$ reachable from $\ell$ through a route $\tau$, with a length $d_{\tau}^{f}[\ell']$ belonging to the interval $[d_1, d_2]$, and such that $\tau[0]=\ell$ and all its elements with index less than $\tau(\ell')$ satisfy $\varphi_1$. 
 In Figure~\ref{fig:spprop}, we report an example of reachability property, considering f as the $hop$ function described in Example~\ref{ex:distancefunction}. In the graph, the location  $\ell_6$ (meaning the trajectory $\vec{x}$ at time t in position $\ell_6$) satisfies $\device \: \reach{[0,1]}{hop} \: \router$. Indeed, there exists a route $\tau = \ell_6\ell_5$ such that $d_{\tau}^{hop}[1]=1$, where $\tau[0]=\ell_6$, $\tau[1]=\ell_5$, $\tau[1]$ satisfies the red property (it is a router) and all the other elements of the route satisfy the blue property (they are end-devices). 

\noindent{\bf Escape.} \[\fmon( \lserv, \vec{x}, \escape{[d_1, d_2]}{f} \: \varphi, t, \ell)= 
\bigvee_{\tau\in Routes(\lserv(t),\ell)} \linebreak
        \bigvee_{\ell'\in \tau:\left(d_{\lambda(t)}^{f}[\ell,\ell'] \in [d_1, d_2]\right) }
            ~\bigwedge_{i \leq \tau(\ell')} 
                \fmon(  \lserv, \vec{x}, \varphi, t, \tau[i]).\]

\noindent $(\lserv, \vec{x}(\ell,t))$
              satisfies  $\escape{d}{f} \: \varphi$ if and only if there exists a route $\route$ and a location $\ell'\in \route$ such that $\route[0]=\ell$ and $d_\mathcal{S}[\tau[0],\ell']$ belongs to the interval $[d_1, d_2]$, while $\ell'$ and all the elements $\tau[0],...\tau[k-1]$ (with $\route(\ell')=k$) satisfy satisfies $\varphi$.             
In Figure~\ref{fig:spprop}, we report an example of escape property. In the graph,  the location $\ell_{10}$ satisfies $ \escape{[2, \infty]}{hop} \: \neg \device$. Indeed, there exists a route $\tau = \ell_{10}\ell_7\ell_8$ such that $\tau[0]=\ell_{10}$, $\tau[2]=\ell_8$, $d_S^{hop}[\ell_{10},\ell_1]=2$ and $\ell_{10}$, $\ell_7$ and $\ell_8$ do not satisfy the blue property, i.e. they are not end-devices. Note that the route $\ell_{10}\ell_{11}\ell_{16}$ is not a good route to satisfy the property because the distance $d_S^{hop}[\ell_{10},\ell_{16}]=1$. 

\begin{figure}[ht]
{\small
\center
\begin{tikzpicture}
  [scale=.6,auto=left,every node/.style={circle,thick,inner sep=0pt,minimum size=6mm}]
  \node (1) [fill=\blu, draw = black] at (-3,-3) {\wh{1}};
  \node (2) [fill=\blu, draw = black] at ( 1,-2) {\wh{2}};
  \node (3) [fill=\blu, draw = black] at ( 3,-1) {\wh{3}};
  \node (4) [fill=\blu, draw = black] at ( -3, 1) {\wh{4}};
  \node (5) [fill=\rouge, draw = black] at ( 1, 2) {\wh{5}};
  \node (6) [fill=\blu, draw = black] at (-1, 2) {\wh{6}};
  \node (7) [fill=\rouge, draw = black] at (0, 0) {\wh{7}};
  \node (8) [fill=\rouge, draw = black] at (-2,-1) {\wh{8}};
  \node (9) [fill=\rouge, draw = black] at (3,3) {\wh{9}};
  \node (10) [fill=\green, draw = black] at (4,1) {\wh{10}};
  \node (11) [fill=\rouge, draw = black] at (5,0) {\wh{11}};
  \node (12) [fill=\blu, draw = black] at (6,-2) {\wh{12}};
  \node (13) [fill=\blu, draw = black] at (8,1) {\wh{13}};
  \node (14) [fill=\blu, draw = black] at (5,3) {\wh{14}};
  \node (15) [fill=\blu, draw = black] at (8,-1) {\wh{15}};
  \node (16) [fill=\rouge, draw = black] at (6.5,1.8) {\wh{16}};

 \draw [-] (1) -- (8) node[midway] {};
 \draw [-] (2) -- (7) node[midway] {};
 \draw [-] (8) -- (6) node[midway] {};
  \draw [-] (8) -- (7) node[midway] {};
 \draw [-] (7) -- (10) node[midway] {};
 \draw [-] (7) -- (5) node[midway] {};
 \draw [-] (3) -- (10) node[midway] {};
  \draw [-] (6) -- (5) node[midway] {};
 \draw [-] (10) -- (11) node[midway] {};
 \draw [-] (10) -- (9) node[midway] {};
 \draw [-] (11) -- (15) node[midway] {};
 \draw [-] (11) -- (12) node[midway] {};
\draw [-] (9) -- (14) node[midway] {};
\draw [-] (10) -- (14) node[midway] {};
\draw [-] (10) -- (16) node[midway] {};
\draw [-] (11) -- (16) node[midway] {};
\draw [-] (13) -- (16) node[midway] {};
\draw [-] (8) -- (4) node[midway] {};

 \end{tikzpicture}
\caption{
Example of spatial properties. {\bf Reachability:} $\device\: \reach{[0,1]}{hop} \: \router$. {\bf Escape:} $ \escape{[2, \infty]}{hop} \: \neg \device$.
{\bf Somewhere}: $\somewhere{[0, 4] }{hop} \coord$.  {\bf Everywhere}:    $\everywhere{[0, 2] }{hop} \router$.   {\bf Surround:} $ (\coord \vee \router ) \surround{[0,3]}{hop} \: \device$. }
\label{fig:spprop}
}
\end{figure}
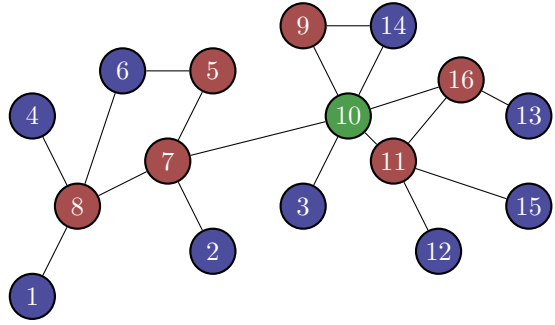

We can also derive  other three spatial operators:  {\it somewhere},  {\it everywhere} and  {\it surround}. 

\noindent{\bf Somewhere.}  $ \somewhere{ [0, d] }{f} \varphi := true  \reach{[0, d]}{f} \varphi $
is satisfied by $(\lserv , \vec{x}(t,\ell))$
iff there exists a location that satisfies $\varphi$ reachable from $\ell$ via a route $\tau$ with a distance  belonging to the interval $[0, d]$. This length is computed via the function $f$. In Figure~\ref{fig:spprop}, all the locations  satisfy the property $\somewhere{[0,3]}{hop} \coord$ because, for all $\ell_i$, there is always a path $\route = \ell_i \dots \ell_{10}$ with a length $d_\route^{hop}(k) \leq 3$, where $\tau[0]=\ell_{i}$, $\tau[k]=\ell_{10}$, and $\ell_{10}$ satisfies the green property (it is a coordinator).

\noindent{\bf Everywhere.}  $ \everywhere{[0,d]}{f} \varphi := \neg \somewhere{[0,d]}{f} \neg \varphi $ 
is satisfied by $(\lserv, \vec{x}(\ell,t))$ iff  all the locations reachable from $\ell$ via a path, with length belonging to the interval $[0, d]$, satisfy $\varphi$. In Figure~\ref{fig:spprop}, there are no locations that satisfy the property $\everywhere{[0,2]}{hop} \router$ because for all the locations $\ell_i$ there is a path $\tau=\ell_i\ell_j$ such that $\ell_j$ does not satisfy the red property (it is not a router).

\noindent{\bf {Surround}.}  $\varphi_{1} \surround{[0, d]}{f} \varphi_{2} := \varphi_{1} \wedge \neg (\varphi_{1}\reach{[0, d]}{f} \neg (\varphi_1 \vee \varphi_{2}) \wedge \neg (\escape{\neg [d, \infty]}{f}  \varphi_{1}) $ expresses the topological notion of being surrounded by a $\varphi_2$-region, while being in a $\varphi_{1}$-region,  with an additional metric constraint. The operator has been introduced in~\cite{CianciaLLM16} as a basic operator, while here it is a derived one. The idea is that one cannot escape from a $\varphi_{1}$-region without passing from a location that satisfies $\varphi_2$ and, in any case, one has to reach a $\varphi_2$-location via a path with a length less or equal to $d$.  In Figure~\ref{fig:spprop}, the location $\ell_{10}$ satisfies the property $ (\coord  \: \vee  \: \router ) \surround{[0,3]}{hop} \: \device$. In fact, it satisfies the green property,  it cannot reach a location that does not satisfy the blue or the red property via a path with length lesser or equal to 3 and it cannot escape through a path satisfying the green or red properties at a distance more than 3. 

The operators can be arbitrarily composed to specify complex properties as we will see in Section~\ref{sec:ZigBee} and~\ref{sec:epidemic}.

\subsection{Invariance properties of the Euclidean spatial model}

The properties we consider with respect to the Euclidean spatial model are typically local and depend on the relative distance and position among nodes in the plane. As such, they should be invariant with respect to the change of coordinates, i.e. with respect to the isometric transformations of the plane. This class of transformations include translations, rotations, and reflections, and can be described by matrix multiplications of the form   
\[
\begin{bmatrix}
    x'_{\ell}      \\
    y'_{\ell}    \\ 
    1              \\ 
\end{bmatrix}
= 
\begin{bmatrix}
       \beta \cos (\alpha)   &   -   \beta \sin (\alpha)  & \beta t_x     \\
       \gamma \sin (\alpha)  &       \gamma \cos (\alpha)  & \gamma t_y \\
    0 &  0 & 1     
\end{bmatrix} 
\begin{bmatrix}
    x_{\ell}      \\
    y_{\ell}    \\ 
    1              \\ 
\end{bmatrix}
\]
where $\beta, \gamma$ are the stretching, $\alpha$ the rotation, and $t_x, t_y$ the translation factor respectively.
Invariance of satisfaction of spatial properties holds in STREL logic, for the Euclidean space model of Definition \ref{def:Euclidean}. Consider an Euclidean space model $\mathcal{E}(L,\mu, R) = \langle L, \wfun^{\mu, R} \rangle$ and $\mathcal{E}(L,\mu', R)= \langle L, \wfun^{\mu', R} \rangle$, obtained by applying an isometric transformation $A$:  $\mu'(\ell) = A(\mu(\ell))$, for invariance to hold, we need to further require that distance predicates used in spatial operators are invariant for isometric transformations. More specifically, for any isometry $A$, we require a distance predicate $d$ on the semiring $\mathbb{R}^{\infty}\times\mathbb{R}^{\infty}$ to satisfy $d((x,y)) = d(A((x,y)))$. This is the case for the norm-based predicates used in the examples,  of the form $d((x,y)) = \|(x,y)\|_2 \leq r$.  

Notice that, the path structure is preserved (the edges given by $R$ is the same), and the truth of isometry-invariant distance predicates along paths in $\mathcal{E}(L,\mu, R)$ and $\mathcal{E}(L,\mu', R)$ is also the same. This straightforwardly implies that the truth value of spatial operators will be unchanged by isometry. 

\begin{proposition}[Equisatisfiability under isometry] Let  $\mathcal{E}(L,\mu, R) = \langle L, \wfun^{\mu, R} \rangle$
be  an Euclidean spatial model and  $\mathcal{E}(L,\mu', R)= \langle L, \wfun^{\mu', R} \rangle$ an isometric transformation of the former. Consider a spatial formula $\varphi_{1} \: \reach{ d}{f} \: \varphi_{2}$ or $\escape{d}{f}  \: \varphi_{1}$, where $d$ is an isometry preserving predicate.
Assume $\fmon ( \lserv, \vec{x}, \varphi_{j}, t, \ell) = \fmon'( \lserv, \vec{x}, \varphi_{j}, t, \ell)$, $j=1,2$, where $\fmon$ and $\fmon'$ are the monitoring functions for the two spatial models. Then it holds that 
$\fmon( \lserv, \vec{x}, \varphi_{1} \: \reach{ d}{f} \: \varphi_{2}, t, \ell) = \fmon'( \lserv, \vec{x}, \varphi_{1} \: \reach{ d}{f} \: \varphi_{2}, t, \ell)$ and $\fmon(\lserv, \vec{x}, \escape{d}{f}  \: \varphi_{1}, t, \ell) = \fmon'( \lserv, \vec{x}, \escape{d}{f}  \: \varphi_{1}, t, \ell)$, for all $\ell$ and $t$. 
\end{proposition}

\section{Monitoring STREL}
\label{sec:alg}
In this section, we present a  monitoring algorithm that can be used to check if a given signal satisfies or not a STREL property. 
The proposed algorithm follows an \emph{offline} approach. Indeed, the proposed algorithm takes as input the complete spatio-temporal signal together with the property we want to monitor. 
%
%
%


\subsection{Offline monitor}
Offline monitoring is performed via a function $\mathsf{monitor}$ 
that takes as inputs a dynamical spatial model $\lserv$, a trace $\vec{x}$ 
and a formula $\phi$ and returns the \emph{piecewise constant 
spatio-temporal signal} $\pcsts$ representing the monitoring 
of $\phi$.
%
The function $\mathsf{monitor}$ is defined by induction on the syntax of the formula (Algorithm \ref{algo:part1}). 
The spatio-temporal signal resulting from the monitoring of atomic proposition $\mu$ is just obtained by applying function $\iota(\mu)$ to the trace $\mathbf{x}$. The spatio-temporal signals associated with $\neg\varphi$ and $\varphi_1\wedge \varphi_2$ are obtained by applying operators $\odot_{D_2}$ and $\otimes_{D_2}$ to the signals resulting from the monitoring of $\varphi$ and from the monitoring of $\varphi_1$ and $\varphi_2$\, where $\oplus_{D_2}$, $\otimes_{D_2}$ and $\odot_{D_{2}}$ depend the \emph{signal domain} used to represent satisfaction values.

\algnewcommand\algorithmicswitch{\textbf{switch}}
\algnewcommand\algorithmiccase{\textbf{case}}
\algnewcommand\algorithmicassert{\texttt{assert}}
\algnewcommand\Assert[1]{\State \algorithmicassert(#1)}%
\algdef{SE}[SWITCH]{Switch}{EndSwitch}[1]{\algorithmicswitch\ #1\ \algorithmicdo}{\algorithmicend\ \algorithmicswitch}%
\algdef{SE}[CASE]{Case}{EndCase}[1]{\algorithmiccase\ #1}{\algorithmicend\ \algorithmiccase}%
\algtext*{EndSwitch}%
\algtext*{EndCase}%

\begin{algorithm}[t]%
  \label{algo:part1}
  \caption{Monitoring algorithm}
  \flushleft \small
\begin{multicols}{2}
\begin{algorithmic}[1]
\Function{Monitor}{$\lserv$, $\vec{x}$, $\psi$}
\Case{$\psi=\nu$}
\State $\pcsts=[]$
\ForAll{ $\ell \in L$ } 
\ForAll{ $t\in \mathcal{T}(\vec{x}(\ell))$ }
\State $\pcsts(\ell,t)=\iota(\mu)(\vec{x}(\ell,t))$
\EndFor
\EndFor
\State \Return{ $\pcsts$ }
\EndCase
\Case{$\psi=\neg\psi_1$}
\State $\pcsts_1=\Call{Monitor}{\lserv,\vec{x},\psi_1}$
\State $\pcsts=[]$
\ForAll{ $\ell \in L$ } 
\ForAll{ $t \in \mathcal{T}(\pcsts_1(\ell))$ } 
\State $\pcsts(\ell,t)=\odot_{D_2} \pcsts_1(\ell,t)$
\EndFor 
\EndFor 
\State \Return{ $\pcsts$ }
\EndCase
\Case{$\psi=\psi_1 \wedge \psi_2$}
\State $\pcsts_1=\Call{Monitor}{\lserv,\vec{x},\psi_1}$
\State $\pcsts_2=\Call{Monitor}{\lserv,\vec{x},\psi_2}$
\State $\pcsts=[]$
\ForAll{ $\ell \in L$ } 
\ForAll{$t \in \mathcal{T}(\pcsts_1(\ell))\cup \mathcal{T}(\pcsts_2(\ell))$} 
\State $\pcsts(\ell,t)=\pcsts_1(\ell,t) \otimes_{D_2} \pcsts_2(\ell,t)$
\EndFor
\EndFor
\State \Return{ $\pcsts$ }
\EndCase
\Case{$\psi=\psi_1 \until{[t1,t2]} \psi_2$}
\State $\pcsts_1=\Call{Monitor}{\lserv,\vec{x},\psi_1}$
\State $\pcsts_2=\Call{Monitor}{\lserv,\vec{x},\psi_2}$
\State $\pcsts=[]$
\ForAll{ $\ell \in L$ } 
\State $\pcsts(\ell)=\Call{Until}{t1,t2,\pcsts_1(\ell),\pcsts_2(\ell)}$
\EndFor
\State \Return{ $\pcsts$ }
\EndCase
\Case{$\psi=\psi_1 \since{[t1,t2]} \psi_2$}
\State $\pcsts_1=\Call{Monitor}{\lserv,\vec{x},\psi_1}$
\State $\pcsts_2=\Call{Monitor}{\lserv,\vec{x},\psi_2}$
\State $\pcsts=[]$
\ForAll{ $\ell \in L$ } 
\State $\pcsts(\ell)=\Call{Since}{t1,t2,\pcsts_1(\ell),\pcsts_2(\ell)}$
\EndFor
\State \Return{ $\pcsts$ }
\EndCase
\Case{$\psi=\psi_1 \reach{[d_1,d_2]}{f} \psi_2$}
\State $\pcsts_1=\Call{Monitor}{\lserv,\vec{x},\psi_1}$
\State $\pcsts_2=\Call{Monitor}{\lserv,\vec{x},\psi_2}$
\State $\pcsts=[]$
\ForAll{ $t \in \mathcal{T}(\pcsts_1)\cup\mathcal{T}(\pcsts_2)$ } 
\State \rlap{$\pcsts@t=\Call{Reach}{\lserv(t),f,d_1,d_2,\pcsts_1@t,\pcsts_2@t}$}
\EndFor
\State \Return{ $\pcsts$ }
\EndCase
\Case{$\psi=\psi_1 \escape{d}{f} \psi_2$}
\State $\pcsts_1=\Call{Monitor}{\lserv,\vec{x},\psi_1}$
\State $\pcsts_2=\Call{Monitor}{\lserv,\vec{x},\psi_2}$
\State $\pcsts=[]$
\ForAll{ $t \in \mathcal{T}(\pcsts_1)\cup\mathcal{T}(\pcsts_2)$ } 
\State \rlap{$\pcsts@t=\Call{Escape}{\lserv(t), f, d,\pcsts_1@t,\pcsts_2@t}$}
\EndFor
\State \Return{ $\pcsts$ }
\EndCase
\EndFunction
\end{algorithmic}
\end{multicols}%
\end{algorithm}

Monitoring of temporal properties, namely $\varphi_1 \until{[t_{1}, t_{2}]}\varphi_2$ and $\varphi_1 \since{[t_{1}, t_{2}]} \varphi_2$, relies on functions \textproc{Until} and \textproc{Since}. These are defined  by using the same approach of~\cite{Donze2013} and~\cite{MalerN13}. However, while their monitoring relies on classical Boolean and arithmetic operators, here the procedure is parametrised with respect to operators $\oplus_{D_2}$ and $\otimes_{D_2}$ of the considered semiring.

\begin{algorithm}[tbp] 
\caption{Monitoring function for \emph{reach} operator}
\label{algo:reachmonitoring}
\vspace{1mm}
\begin{algorithmic}[1]
\Function{Reach}{$(L,\wfun)$, $f: A\rightarrow B$, $d_1\in B$, $d_2\in B$, $\ssign_1: L\rightarrow D_2$, $\ssign_2:L\rightarrow D_2$}
\If{$d_2\not=\top_{B}$}
\State{} \Return \Call{BoundedReach}{$(L,\wfun)$, $f$, $d_1$, $d_2$ , $\ssign_1$, $\ssign_2$}
\Else{}
\State{} \Return \Call{UnboundedReach}{$(L,\wfun)$, $f$, $d_1$, $\ssign_1$, $\ssign_2$}
\EndIf
\EndFunction
\end{algorithmic}
\end{algorithm}

To monitor $\varphi_{1} \: \reach{[d_{1},d_{2}]}{f:A\rightarrow B} \: \varphi_{2}$
 we first compute 
the signals $\mathbf{s}_1$ and $\mathbf{s}_2$, resulting from the monitoring of $\varphi_1$ and $\varphi_2$.  After that, the final result is obtained by aggregating the spatial signals $\mathbf{s}_1@t$ and $\mathbf{s}_2@t$ at each time $t\in \mathcal{T}(\mathbf{s}_1)\cup \mathcal{T}(\mathbf{s}_2)$ by using function \textproc{Reach} defined in Algorithm~\ref{algo:reachmonitoring}. In this function two cases are distinguished: $d_2\not=\top_{B}$ or $d_2=\top_{B}$. In the first case, the resulting monitoring value is calculated via function \textproc{BoundedReach} defined in Algorithm~\ref{algo:boundedreachmonitoring}. Conversely, when $d_2=\top_{B}$ monitoring is performed by relying on function \textproc{UnboundedReach} defined in Algorithm~\ref{algo:unboundedreachmonitoring}.

\noindent \paragraph{\bf  \textproc{BoundedReach}} Function \textproc{BoundedReach}, defined in Algorithm~\ref{algo:boundedreachmonitoring}, takes as parameters the spatial model $\langle L,\wfun \rangle$ at time $t$, the function $f:A\rightarrow B$, used to compute the distances over paths, and the interval $[d_1,d_2]$, describing the reachability bound.
The proposed algorithm is based on flooding that computes the output signal $\ssign$. 
At line $3$, we inizialite the set $Q$, it contains all triples  $(\ell, s_2[\ell], \bot_B)$, where $\bot_B$ is the minimum element of the distance domain $B$ (e.g. if $B=\mathbb{R}_{\geq 0}$, $\bot_B=0$).
Let us denote $Q_i$ the value of  $Q$ after $i$ iterations of loop starting at line $4$.
$Q_i$ contains a triple  $(\ell, v, d)$ if and only if there exists a path such that with  $i$-steps we reach a distance $d <_{B} d_2$ and a monitored value $v$. 
To compute the values in $Q_{i+1}$, for each element $(\ell,v,d)\in Q_i$, we consider the locations $\ell'$ next to $\ell$  at a distance $w$ ($\nextto{\ell'}{w}{\ell}$) and we compute the items: $v' = v\otimes \ssign_1(\ell')$ and $d' = d+_{B}f(w)$. The element $(\ell', v',d')$ is added to $Q_{i+1}$ if $d'<_{B} d_2$, i.e. if the sum of the current distance plus the distance between $\ell$ and $\ell'$ is still less than $d_2$.
When  $d' \in [d_1,d_2]$, $\ssign(\ell')$ is updated to take into account the new computed value. We recal that $s$ stores the value of the semantics of the reach operator.

\begin{algorithm}[tbp] 
\caption{Monitoring bounded reachability}
\label{algo:boundedreachmonitoring}
\vspace{1mm}
\begin{algorithmic}[1]
\Function{BoundedReach}{$(L,\wfun)$, $f: A\rightarrow B$, $d_1\in B$, $d_2\in B$, $\ssign_1: L\rightarrow D_2$, $\ssign_2:L\rightarrow D_2$} 
\State $\forall \ell\in L. \ssign[\ell] = \left\{\begin{array}{ll}
\ssign_2[\ell] & d_1=\bot_{B}\\
\bot_{D_2} & \mbox{otherwise}
\end{array}\right.
$
\State $Q=\{ (\ell, \ssign_2[\ell], \bot_{B}) | \ell\in L \}$
\While{ $Q\not=\emptyset$ }
\State{$Q'=\emptyset$}
\ForAll{$(\ell,v,d) \in Q$} 
\ForAll{ $\ell': \nextto{\ell'}{w}{\ell}$ }
\State $v'=v~\otimes_{D_2}~\ssign_1[\ell']$
\State $d'=d~+_{B}~f(w)$
\If{$(d_1\leq d'\leq d_2)$}
\State $\ssign[\ell'] = \ssign[\ell']\oplus_{D_2} v'$
\EndIf
\If{$d'< d_2$}
\If{$\exists (\ell',v'',d')\in Q'$} 
\State $Q' = (Q'-\{ (\ell',v'',d') \})\cup \{ (\ell',v'\oplus_{D_2} v'',d') \}$
\Else{}
\State $Q' = Q'\cup \{ (\ell',v',d') \}$
\EndIf
\EndIf
\EndFor
\EndFor
\State{$Q=Q'$}
\EndWhile
\State \Return{ $\ssign$}
\EndFunction
\end{algorithmic}
\end{algorithm}

\noindent \paragraph{\bf  \textproc{UnboundedReach}} Function  \textproc{UnboundedReach} defined in Algorithm~\ref{algo:unboundedreachmonitoring} is used when the interval in the \emph{reach} formula is unbounded. In this case the function takes as parameters the spatial model $\langle L,\wfun\rangle $ at time $t$, the function $f:A\rightarrow B$, used to compute the distances over paths, and the lower bound $d_1$, and returns a \emph{spatial signal} $\ssign$. If  $d_1 = \bot_B$, i.e. when we are considering a totally unbound reach with no constraints, we initialize $\ssign=\ssign_2$. Otherwise, when $d_1\not=\bot$, we have first to call function \textproc{boundedReach} by passing as upper bound $d_1+d_{max}$, where $d_{max}$ is the max value that function $f$ can associate to a single edge in $\wfun$.
After that, $\ssign$ will contain the \emph{reachability} value computed up to the bound $[d_1,d_1+d_{max}]$ (lines (5)-(6)). 
Hence, the computed values are \emph{back propagated} until a fixpoint is reached. 
This will guarantee that for each location, only the values of $\ssign_2$ at a path distance $[d_1,\top_{B}]$ are considered in the computation of reachability values.

%


\begin{algorithm}[tbp] 
\caption{Monitoring unbounded reachability}
\label{algo:unboundedreachmonitoring}
\vspace{1mm}
\begin{algorithmic}[1]
\Function{UnboundedReach}{$(L,\wfun)$, $f: A\rightarrow B$, $d_1\in B$, $\ssign_1: L\rightarrow D_2$, $\ssign_2:L\rightarrow D_2$}
\If{$d_1=\bot_{B}$}
\State $\ssign=\ssign_2$
\Else
\State $d_{max}=\max\{ f(w) | \exists.\ell,\ell'\in L: \nextto{\ell}{w}{\ell'} \}$
\State $\ssign=\Call{BoundedReach}{(L,\wfun), f, d_1, d_1+_{B} d_{max},\ssign_1,\ssign_2}$
\EndIf
\State $T=L$
\While{$T\not=\emptyset$}
\State $T'=\emptyset$
\ForAll{$\ell\in T$}
\ForAll{ $\ell': \nextto{\ell'}{w}{\ell}$ }
\State $v' = (\ssign[\ell]\otimes_{D_2} \ssign_1[\ell'])\oplus_{D_2} \ssign[\ell']$
\If{$v'\not= \ssign[\ell']$}
\State{$\ssign[\ell']=v'$}
\State $T'=T'\cup\{ \ell' \}$
\EndIf
\EndFor
\EndFor
\State $T=T'$
\EndWhile
\State \Return{ $\ssign$}
\EndFunction
\end{algorithmic}
\end{algorithm} 

\noindent \paragraph{\bf  \textproc{Escape}} Monitoring algorithm for $\escape{[d_{1},d_{2}]}{f:A\rightarrow B}{\varphi}$ is reported in Algorithm~\ref{algo:escapemonitoring}. Given a  spatial model $\langle L,\wfun\rangle $ at time $t$, a distance function $f:A\rightarrow B$, an interval $[d_1,d_2]$, it computes the spatial signal representing the monitoring value of $\escape{d}{f} \varphi$ at time $t$. 
Function $\mathsf{escape}$ first computes the \emph{matrix distance} $D$ (line 2), obtained from the given space model and distance function $f$. After that, a matrix $e:L\times L\rightarrow D_2$ is computed. 
The matrix $e$ is initialised so that all the elements $e[\ell,\ell]$ in the diagonal are equal to $\ssign_1(\ell)$, while all the other elements are set to $\bot_{D_2}$ (lines 3-4).
After that, iteratively, elements of $e$ are updated by considering the values in the neighbours in each location (lines 6-20). A value $e[\ell_1',\ell_2]$ is updated iff  
$\ssign_1(\ell_1') \otimes_{D_2} e[\ell_1,\ell_2] >_{D_2}  e[\ell_1',\ell_2] $, where $\ell_1'$ is a neighbor of $\ell_1$.
The updates end when a fixpoint is reached.\footnote{We prove that the loop always terminates in Lemma~\ref{lemma:escapecorrectness}.}
At the end of the loop computation, the element $e[\ell_1,\ell_2]$ contains the \emph{escape} value from $\ell_1$ to $\ell_2$, defined in the semantics without the distance constraint. This latter is took into consideration in line 23, where the final monitored value $s$ is computed. For each $\ell$, the equation $\bigoplus_{D_2}(\{ e[\ell,\ell'] | D[\ell,\ell']\in [d_1,d_2] \})$ considers the minimum values $e[\ell,\ell']$ of all $\ell'$ that satisfies the distance contraint, i.e. such that $ D[\ell,\ell']\in [d_1,d_2]$.



\begin{algorithm}[tbp] 
\caption{Monitoring \emph{escape}}
\label{algo:escapemonitoring}
\vspace{1mm}
\begin{algorithmic}[1]
\Function{Escape}{$(L,\wfun)$, $f: A\rightarrow B$, $d_1\in B$,$d_2\in B$, $\ssign_1: L\rightarrow D_2$} 
\State $D = \Call{MinDistance}{L,\wfun,f)}$
\State $\forall \ell,\ell'\in L. e[\ell,\ell'] = \bot_{D_2}$
\State $\forall \ell\in L. e[\ell,\ell]=\ssign_1(\ell)$
\State $T=\{ (\ell,\ell) | \ell\in L \}$
\While{ $T\not= \emptyset$ }
\State $e'=e$ 
\State $T'=\emptyset$
\ForAll{ $(\ell_1,\ell_2) \in T$ }
\ForAll{ $\ell_1': \nextto{\ell_1'}{w}{\ell_1}$ }
\State{ $v = e[\ell_1',\ell_2]\oplus_{D_2}(\ssign_1(\ell_1') \otimes_{D_2} e[\ell_1,\ell_2])$}
\If{$v\not=e[\ell_1',\ell_2]$}
\State{$T'=T'\cup \{ (\ell_1',\ell_2) \}$}
\State{$e'[\ell_1',\ell_2]=v$}
\EndIf
\EndFor
\EndFor
\State{T=T'}
\State{e=e'}
\EndWhile
\State $\ssign=[]$
\ForAll{ $\ell\in L$ }
\State $\ssign(\ell)=\bigoplus_{D_2}(\{ e[\ell,\ell'] | D[\ell,\ell']\in [d_1,d_2] \})$
\EndFor{}
\State \Return $\ssign$
\EndFunction
\end{algorithmic}
\end{algorithm}

\subsection{Correctness}
In this subsection, we discuss the correctness of the algorithms.
\begin{lemma}[BoundedReach correctness]
\label{lemma:boundreachcorrectness)}
Given an $A$-spatial model $\lserv=(L,\wfun)$, a function $f: A\rightarrow B$, an interval $[d_1, d_2]$ (with $d_1,d_2\in B$, $d_1\leq_{B} d_2$ and $d_2\not=\top_{B}$), and two spatial signals $\ssign_1: L\rightarrow D_2$, $\ssign_2:L\rightarrow D_2$ such that $\Call{BoundedReach}{(L,\wfun),f,d_1,d_2,\ssign_1,\ssign_2}=\ssign$, for any $\ell\in L$, we have that:

\[
\ssign(\ell)={\oplus_{D_2}}_{\tau\in Routes(\lserv,\ell)}
        ~~{ \oplus_{D_2}}_{i : \left(d_{\tau}^{f}[i] \in [d_{1},d_{2}]\right)}
        \left(
            \ssign_2(\tau[i])
            \otimes_{D_{2}}
            {\otimes_{D_2}}_{j < i} 
                \ssign_1(\tau[j])     
        \right)
\]

\end{lemma}

\begin{proof}
Let us denote by $\ssign_{i}$ and $Q_i$ the value of variables $\ssign$ and $Q$ respectively in Algorithm~\ref{algo:boundedreachmonitoring} after $i$ iterations of the \emph{while-loop} at line $(4)$. Considering that, since $f(w)>0$ for any $w\in A$, the algorithm terminates after a finite number of iterations, the statement follows directly from the following properties:
\begin{enumerate}[align=left]
    \item[$P1$:] if $(\ell,v,d)\in Q_i$ then $d\leq_{B} d_2$;
    \item[$P2$:] if $(\ell,v_1,d)\in Q_i$ and $(\ell,v_2,d)\in Q_i$ then $v_1=v_2$;
    \item[$P3$:] $(\ell,v,d)\in Q_i$ if and only if 
    \[ 
    v={\oplus_{D_2}}_{\tau\in Routes(\lserv,\ell): d_{\tau}^{f}[i]=d}
            \ssign_2(\tau[i])
            \otimes_{D_{2}}
        \left(
            {\otimes_{D_2}}_{j < i} 
                \ssign_1(\tau[j])             
        \right)
\]
\item[$P4$:] for any $\ell\in L$:
\[
\ssign_i[\ell]={\oplus_{D_2}}_{\tau\in Routes(\lserv,\ell)}
        ~~{ \oplus_{D_2}}_{k : k\leq i\wedge \left(d_{\tau}^{f}[k] \in [d_{1},d_{2}]\right)}
        \left(
            \ssign_2(\tau[k])
            \otimes_{D_{2}}
            {\otimes_{D_2}}_{j < k} 
                \ssign_1(\tau[j])             
        \right)
\]
\end{enumerate}

Properties $P1$ and $P2$ are direct consequences of instructions at lines $(3)$ and $(13)-(19)$ of Algorithm~\ref{algo:boundedreachmonitoring}. Indeed, $(\ell,v,d)\in Q_0$ if and only if $d=\bot_{B}$ (line $(3)$), while $(\ell,v,d)$ is inserted in $Q'=Q_{i+1}$ if and only if $d<_{B} d_2$ (line $(13)$) and no other $(\ell,v',d)$ is already in $Q'$ (lines $(14)-(18)$).

Property $P3$ can be proved by induction on $i$ by observing that the property is satisfied by $Q_0$ and that for any $i$:

\[
(\ell',v',d')\in Q_{i+1} \Leftrightarrow d'<_{B} d_2\mbox{ and }
v'={\oplus_{D_2}}_{(\ell,v,d)\in Q_i:\nextto{\ell'}{w}{\ell}\wedge d+f(w)=d'} 
        \left(
            \ssign_1(\ell')
            \otimes_{D_2}
            v
        \right)
\]

From the above, and from inductive hypothesis, we have that:
\[ \small
\begin{array}{r@{\,}c@{\,}l}
v' &=& {\oplus_{D_2}}_{(\ell,v,d)\in Q_i:\nextto{\ell'}{w}{\ell}\wedge d+f(w)=d'} 
        \left(
            \ssign_1(\ell')
            \otimes_{D_2}
            {\oplus_{D_2}}_{\tau\in Routes(\lserv,\ell): d_{\tau}^{f}[i]=d}
            \ssign_2(\tau[i])
            \otimes_{D_{2}}
        \left(
            {\otimes_{D_2}}_{j < i} 
                \ssign_1(\tau[j])             
        \right)
        \right)\\
        & = & 
            {\oplus_{D_2}}_{\tau\in Routes(\lserv,\ell'): d_{\tau}^{f}[i+1]=d'}
            \ssign_2(\tau[i+1])
            \otimes_{D_{2}}
        \left(
            {\otimes_{D_2}}_{j < i+1} 
                \ssign_1(\tau[j])             
        \right)
\end{array}
\]
That is the statement of $P3$.

Finally, we can probe $P4$ by induction on $i$ by using $P3$ and by observing that:
\[
\ssign_{i+1}[\ell']=
{\oplus_{D_2}}_{(\ell,v,d)\in Q_i:\nextto{\ell'}{w}{\ell}\wedge d+f(w)\in [d_1,d_2]} \ssign_{i}[\ell']\oplus_{D_2}\left(
\ssign_1[\ell'] \otimes v\right)
\tag*{\qedhere}
\]
\end{proof}

\begin{lemma}[UnboundedReach correctness]
\label{lemma:unboundreachcorrectness}
Given an $A$-spatial model $(L,\wfun)$, a function $f: A\rightarrow B$, a value $d_1\in B$ ($d_1\not= \top_{B}$), and two spatial signals $\ssign_1: L\rightarrow D_2$, $\ssign_2:L\rightarrow D_2$ such that $\Call{UnboundedReach}{(L,\wfun),f,d_1,\ssign_1,\ssign_2}=\ssign$, for any $\ell\in L$, we have that:

\[
\ssign(\ell)={\oplus_{D_2}}_{\tau\in Routes(\lserv(t),\ell)}
        ~~{ \oplus_{D_2}}_{\ell'\in\tau : \left(d_{\tau}^{f}[\ell'] \geq  d_{1}\right)}
        \left(
            \ssign_2(\ell')
            \otimes_{D_{2}}
            {\otimes_{D_2}}_{j < \tau(\ell')} 
                \ssign_1(\tau[j])             
        \right)
\]

\end{lemma}

\begin{proof}
Directly from the pseudo code in Algorithm~\ref{algo:unboundedreachmonitoring} and from Lemma~\ref{lemma:unboundreachcorrectness}, we can observe that the value $\ssign$ computed by function $\Call{UnboundedReach}{}$ is the limit ($\ssign = \lim_{i\rightarrow\infty }\ssign_{i}
$) of the sequence of signals $\ssign_i$ such that for any $\ell\in L$:
\[
\ssign_{i+1}[\ell]=\bigoplus_{\ell\in L:\nextto{\ell}{w}{\ell'}} (\ssign_{i}(\ell)\oplus \ssign_1(\ell')\otimes \ssign_{i}(\ell'))
\]
The initial spatial signal is $\ssign_{0}=\ssign_2$, if $d_1=\bot_{B}$, while it is:
\[
\ssign_0[\ell]={\oplus_{D_2}}_{\tau\in Routes(\lserv,\ell)}
        ~~{ \oplus_{D_2}}_{i : \left(d_{\tau}^{f}[i] \in [d_{1},d_{1}+d_{max}]\right)}
        \left(
            \ssign_2(\tau[i])
            \otimes_{D_{2}}
            {\otimes_{D_2}}_{j < i} 
                \ssign_1(\tau[j])             
        \right)
\]
\noindent
when $d_1\not=\bot_{B}$ and $d_{max}=\max\{ f(w) | \exists.\ell,\ell'\in L: \nextto{\ell}{w}{\ell'} \}$. In both the cases, the statement follows by applying standard algebra and the properties of $\oplus$ and $\otimes$.
\end{proof}

\begin{lemma}[Escape correctness]
\label{lemma:escapecorrectness}
Given an $A$-spatial model $(L,\wfun)$, a function $f: A\rightarrow B$, an interval $[d_1, d_2]$ (with $d_1,d_2\in B$, $d_1\leq_{B} d_2$ and) $d_2\not=\top_{B}$), and a spatial signal $\ssign_1: L\rightarrow D_2$ such that $\Call{Escape}{(L,\wfun),f,d_1,d_2,\ssign_1}=\ssign$, for any $\ell\in L$, we have that:

\[
\ssign(\ell)={\bigoplus_{D_2}}_{\tau\in Routes((L,\wfun),\ell)}
        ~~{\bigoplus_{D_2}}_{\ell'\in \tau : \left(d_{(L,\wfun)}^{f}[\ell,\ell'] \in [d_{1},d_{2}]\right)}
            ~~{\bigotimes_{D_2}}_{i \leq \tau(\ell')} \ssign_1(\tau[i])
\]

\begin{proof}
Let us denote by $e_{i}$ the content of data structures $e$ after $i$ iterations of the loop at line $(6)$ of Algorithm~\ref{algo:escapemonitoring}.
We have only to prove the following properties:

\begin{itemize}[align=left]
    \item[$P1$] For any $\ell_1,\ell_2$, $D[\ell_1,\ell_2]=d_{(L,\wfun)}^{f}[\ell,\ell']$ 
    \item[$P2$] For any $i$:

\[
e_{i}[\ell_1,\ell_2]={\bigoplus_{D_2}}_{\tau\in Routes((L,\wfun),\ell)}
            ~~{\bigotimes_{D_2}}_{j \leq \tau(\ell')\wedge j\leq i} \ssign_1(\tau[j])
\]

    \item[$P3$] The loop terminates after at most $k=|L|$ iterations and 
    
\[
e_{k}[\ell_1,\ell_2]={\bigoplus_{D_2}}_{\tau\in Routes(\lserv(t),\ell)}
            ~~{\bigotimes_{D_2}}_{j \leq \tau(\ell')} \ssign_1(\tau[j])
\]

\end{itemize}

Property $P1$ follows directly from definition of $d_{(L,\wfun)}^{f}[\ell,\ell']$ and from the fact that $\mathsf{MinDistance(L,\wfun,f)}$ computes the matrix of min distances computed in $\langle L,\wfun\rangle$ via $f$. Property $P2$ can be proved by induction on $i$  and follows directly from the code of Algorithm~\ref{algo:escapemonitoring}. Finally, $P3$ is a consequence of the fact that after at most $|L|$ iterations a fix point is reached since all the locations have been taken into account.  
We can conclude that the statement of the lemma follows directly from properties $P1$, $P2$ and $P3$ above by observing that:

\[
\begin{array}[b]{rcl}
\ssign(\ell) & = & \bigoplus_{D_2}(\{ e[\ell,\ell'] | D[\ell,\ell']\in [d_1,d_2] \}) \\
& = & \bigoplus_{D_2}(\{ e_{k}[\ell,\ell'] | D[\ell,\ell']\in [d_1,d_2] \}) \\
& = & {\bigoplus_{D_2}}_{\tau\in Routes(\lserv(t),\ell)}
    ~~{\bigoplus_{D_2}}_{\ell'\in \tau : \left(d_{(L,\wfun)}^{f}[\ell,\ell'] \in [d_{1},d_{2}]\right)}
            ~~{\bigotimes_{D_2}}_{j \leq \tau(\ell')} \ssign_1(\tau[j])
\end{array} 
\tag*{\qedhere}
\]

\end{proof}
\end{lemma}

\begin{theorem}
Given a dynamical spatial model $\lserv$, a trace $\vec{x}$ 
and a formula $\phi$, we have that:

\[
\Call{Monitor}{\lserv,\vec{x},\phi}=\pcsts_{\phi}^{\lserv,\vec{x}}
\]

\end{theorem}

\begin{proof}
The proof easily follows by induction on $\phi$ by using Lemma~\ref{lemma:boundreachcorrectness)},
Lemma~\ref{lemma:unboundreachcorrectness}, and 
Lemma~\ref{lemma:escapecorrectness}.
\end{proof}

\subsection{Complexity}
In this subsection, we discuss the complexity of each algorithm.

\begin{proposition}[BoundedReach complexity]
\label{prop:reachboundcomplexity}
Given an $A$-spatial model $\lserv=(L,\wfun)$, a function $f: A\rightarrow B$, an interval $[d_1, d_2]$ (with $d_1,d_2\in B$, $d_1\leq_{B} d_2$ and) $d_2\not=\top_{B}$), and two spatial signals $\ssign_1: L\rightarrow D_2$, $\ssign_2:L\rightarrow D_2$ such that $\Call{BoundedReach}{(L,\wfun),f,d_1,d_2,\ssign_1,\ssign_2}=\ssign$, we define $d_{min}=\min_{(\ell, w, \ell') \in \wfun}(f(w))$ and $k=\min\{ i | i*d_{min} > d_{2} \}$, where $i*d_{min}$ indicates the sum of $i$ copies of $d_{min}$,  then the algorithm terminates after $O(k\cdot \beta_{d_2}\cdot|L|\cdot m)$ steps, where $m$ is the number of edges and $\beta_{d_2}$ is an integer counting the \emph{different distances} accumulated after $k$ steps.\footnote{ This value is in practice a constant and depends on the weights associated with edges and on the bound $d_2$. For instance, in the case of function $hop$ of Example~\ref{ex:distancefunction}, $\beta=1$. In general, $\beta_{d_2}$ has the same order of magnitude as $k$ considered above.}
\end{proposition}
\begin{proof}
First, we need to compute the upper bound on the number of iterations of the while loop starting at the line (4).
Let us denote by $Q_k$ the value of $Q$ after $k$ iterations. If $Q_k=\emptyset$, the loop stops after at most $k$ iterations.  $Q_k$ is empty if no elements are added at that iteration. 
An element $(\ell', v', d')$ is not added to $Q_k$ iff $d' \geq d_2$ where $d'=d +_{B} f(w)\geq d +_{B} d_{min} $  . Note that, at the first iteration $Q_0$, $d = \bot_B$. At each iteration, we add a value greater or equal to $d_{min}$, this means that after $k$ iterations  $d' \geq k*d_{min}$ but  $k*d_{min} > d_2$ for definition.  Q can have at most $\beta_{d_2}\cdot | L |$ elements and, at each iteration of the while loop, for each elements in Q, we consider their connected ones.
This implies that  function \Call{BoundedReach}{} terminates after $O(k\cdot \beta_{d_2}\cdot|L|\cdot m)$ steps. 
\end{proof}

\begin{proposition}[UnBoundedReach complexity]
\label{prop:reachunboundcomplexity}
Given an $A$-spatial model $(L,\wfun)$, a function $f: A\rightarrow B$, a value $d_1\in B$ ($d_1\not= \top_{B}$), and two spatial signals $\ssign_1: L\rightarrow D_2$, $\ssign_2:L\rightarrow D_2$ such that $\Call{UnboundedReach}{(L,\wfun),f,d_1,\ssign_1,\ssign_2}=\ssign$ and let $d_{min}=\min_{(\ell, w, \ell') \in \wfun}(f(w))$, $d_{max}=\max_{(\ell, w, \ell') \in \wfun}(f(w))$ and $k=\min\{ i | i*d_{min} > d_{1}+d_{max} \}$, then the algorithm stops after $O(k\cdot \beta_{d_2}\cdot|L|\cdot m)$ steps, where $m$ is the number edges and $\beta_{d_2}$ is an integer counting the \emph{different distances} accumulated after $k$ steps. Furthermore, when $d_1=\bot_{B}$, this complexity reduces to $O(|L|*m)$.
\end{proposition}




\begin{proof}
We have already observed in Proposition~\ref{prop:reachboundcomplexity} that the first part of Algorithm~\ref{algo:unboundedreachmonitoring} terminates after $O(k\cdot \beta_{d_2}\cdot|L|\cdot m)$. We can here observe that the second part of the algorithm (lines $(9)-(21)$) needs at most $|L|\cdot m$ steps. This is because the for-loop at lines $(12)-(18)$ consists of at most $O(m)$ steps (indeed each edge is considered at most two times). Moreover, a location can be inserted in $T$ at most $|L|$ times.
Concluding, the algorithm terminates after $O(k\cdot \beta_{d_2}\cdot|L|\cdot m)+O(|L|\cdot m)=O(k\cdot \beta_{d_2}\cdot|L|\cdot m)$ steps. When $d_1=\bot_{B}$  lines $(9)-(21)$ are not executed, then the algorithm 
terminates after $O(|L|\cdot m)$.
\end{proof}


\begin{proposition}[Escape complexity]
\label{prop:escapecomplexity}
Given an $A$-spatial model $(L,\wfun)$, a function $f: A\rightarrow B$, an interval $[d_1, d_2]$, ($d_1,d_2\in B$, $d_1\leq_{B} d_2$ and) $d_2\not=\top_{B}$), and a spatial signal $\ssign_1: L\rightarrow D_2$ such that $\Call{Escape}{(L,\wfun),f,d_1,d_2,\ssign_1}=\ssign$. Algorithm terminates in at most  $O(|L|\cdot m )$ steps, where $m$ is the number of edges.
\end{proposition}

\begin{proof}
The computation of function $\Call{MinDistance}{L,\wfun,f}$ needs $O( m \log(|L|))$ steps.
Moreover, from property $P3$ in the proof of Lemma~\ref{lemma:escapecorrectness}, we have that the loop at the line (6) terminates after at most $|L|$ iterations. In each of these iterations, each edge is taken into account at most 2 times (one for each of the connected locations). This means that the loop terminates after $O(|L|*m)$ steps. 
Finally, to compute the resulting spatial signal, $O(|L|)$ steps are needed for the loop at line $(22)$.
Concluding, the algorithm terminates in $O( m \log(|L|))+O(|L|*m)+O(|L|)$, that is $O(|L|*m)$.
\end{proof}

We can conclude this section by observing that the number of steps needed to evaluate function $\Call{Monitor}{}$ in Algorithm~\ref{algo:part1} is linear in the size of $\phi$, in the length of the signal, and in the number of \emph{edges} in the spatial model and it is quadratic in the number of locations.

\section{Case study: ZigBee protocol monitoring}
\label{sec:ZigBee}
In this section, we consider the running example used in the previous sections. We discuss some properties to show the expressivity and potentiality of STREL.

Given a MANET with a ZigBee protocol (Example~\ref{ex:zigbee}), 
we consider as spatial models both its proximity and connectivity graphs computed with respect to the Cartesian coordinates (Example \ref{ex:manet}). Nodes have three kind of roles: {\it coordinator}, {\it router} and {\it EndDevice}, as described in Example \ref{ex:zigbee}. Moreover, each device is also equipped with a sensor to monitor its battery level ($b$), the humidity ($h$) and the pollution  ($p$) in its position.  
The semiring is the union between the  \emph{max/min} semiring $\mathbb{R}^{\infty}$ (for the proximity graph) and the  \emph{integer} semiring $\mathbb{N}^{\infty}$ (for the connectivity graph). We will use also two types of distances: ${\it hop}$ and the $\Delta$ distances described in Example~\ref{ex:distancefunction}.
Atomic propositions  $\{ \coord, \router, \device\}$ describe the type of nodes. We also consider inequalities on the values that are read from sensors, plus special propositions $@_\ell$ which encode the address of a specific location, i.e. they are true only in the location $\ell$. 

In the following, we describe several properties of these ZigBee MANET networks that are easily captured by STREL logic, to exemplify its expressive power.

A class of properties naturally encoded in STREL is related to the connectivity of the network. First, we can be interested to know if a node is properly connected, meaning that it can reach the coordinator through a path of routers:
\begin{equation}
\phi_{connect} = device \reach{[0,1]}{hop} (router \reach{ }{hop} coord )
\end{equation}
The meaning of this property is that one end node reaches in a step a node which is a router and that is connected to the coordinator via a path of routers.

We may also want to know if there is a path to the router which is reliable in terms of battery levels, for instance such that all routers have a battery level, $b$, above 50\%:  
\begin{eqnarray}
&\phi_{reliable\_router} = ((b > 0.5) \wedge router) \reach{ }{hop} coord &\nonumber \\
&\phi_{reliable\_connect} =   device \reach{[0,1] }{hop} (\phi_{reliable\_router}  )&
\end{eqnarray}

The properties focus on spatial connectivity at a fixed time. We can add also temporal requirements, for instance asking that a broken connection is restored within $h$ time units:

\begin{equation}
\phi_{connect\_restore} = \glob{} (\neg \phi_{connect} \rightarrow  \ev{[0,h]}\phi_{connect} )
\end{equation}

Another class of properties of interest is  the acyclicity of transmissions. To this end, we need to force the connectivity graph to be directed, with edges pointing in the direction of the coordinator (i.e. transmission reduces the distance from the coordinator).  With STREL, we can easily detect the  absence of a cycle for a fixed location $\ell$. This is captured by 
$\phi^{\ell}_{acyclic} = \neg \phi^{\ell}_{cycle}$, where 
\begin{equation}
\phi^{\ell}_{cycle} =  @_\ell \reach{[0,1]}{hop}  (\neg  @_\ell  \wedge \somewhere{}{hop}@_\ell)
\end{equation}
In order to characterize the whole network as acyclic, we need to take the conjunction of the previous formulae for all locations (or at least for routers, enforcing end devices to be connected only with routers). This is necessary as STREL is interpreted locally, on each location, and this forbids us to express properties of the whole network with location unaware formulae. This is a price for efficient monitoring, as global properties of networks require more expressive and computationally expensive logic. 
However, we can use the parametrization of STREL and the property of a Voronoi diagram to specify the global connection or the acyclicity of the graph. Indeed, the proximity graph connects always all the locations of the system, then the property $\everywhere{}{\Delta} \phi$, verified on the proximity graph, holds iff $\phi$ holds in all the locations of the system.
 
Up to now, we have presented qualitative properties, depending on the type of node. If we express properties of sensor measurements, we can also consider quantitative semantics, returning a measure of the robustness of (dis)satisfaction. As an example, we  can monitor \eqref{eq:f1} if in each location a high value of pollution eventually implies,  within $T$ time units, an high value of humidity, or \eqref{eq:f2} in which locations it is possible to find a `safe' route, where both the humidity and the pollution are below a certain threshold. We can also check \eqref{eq:f3} if a location is at least at distance at most $d$ from a location that is safe. 
\begin{eqnarray}
&\phi_{PH} = (p > 150) \Rightarrow \ev{[0,T]} (h > 100) \label{eq:f1} &\\
&\phi_{Safe} =\glob{[0,T]} \escape{[d, \infty]}{\Delta} \: {(h < 90) \wedge (p < 150) } \label{eq:f2} &\\
&\phi_{some} = \somewhere{[0,d]}{\Delta}  \phi_{Safe} \label{eq:f3} &
\end{eqnarray}


\section{Case study: epidemic spreading}
\label{sec:epidemic}
In this section, we investigate a case study based on an epidemic spreading model in a population of a disease-transmitting via direct contact, like flu or COVID19. The simplest models of epidemic spreading are based on a mean-field assumption and treat the population as a homogeneous entity, assuming equal probability that two individuals can enter into contact \cite{epidemiology2019}. A more accurate description, instead,  models potential contacts in a more explicit way, hence the population is represented as a network of interacting agents \cite{network_epidemic2015}, in which nodes are the agents and links are the potential contacts. Such a network can be static (links are fixed) or dynamic (links change during the simulation) and possibly adaptive \cite{network_epidemic2015}.
These kinds of models are particularly useful when dealing with scenarios in which the number of potentially infective contacts, and thus of secondary infections, can vary a lot between individuals, the so-called super-spreading scenario \cite{superspreading_2005}, which seems to be the relevant one also to capture the spreading of COVID19 disease \cite{superspreading_COVID_2020}.

In our case study, we consider a discrete-time model composed of two contact networks, one static, describing contacts within  work colleagues, family, closed relatives, and friends, and another one dynamic, modeling less frequent interaction events, like going to the restaurant, the cinema, or the disco.
The static network is defined by a degree distribution, assumed to be the same for each node, and modeled as a lognormal distribution with cutoff (mean 10, 99 percentile 50, cut off at 200).\footnote{Contact distributions are constructed to resemble contact data collected by a regional government in Italy, which is not publicly available.} To generate the network, we sample a degree for each node and then sample the network relying on the \textsc{expected\_degree\_graph} method of networkX Python library \cite{networkX}.
This network is sampled once and not modified during simulations.
The dynamic event network, instead, is resampled at every simulation step (essentially corresponding to a day). Here, we additionally choose a subset of nodes that will be affected by the events. Each node has assigned a random probability of taking part in the event (chosen uniformly among the following frequency: once every month, once every two weeks, once every week, twice per week) and at each step, the node is included in the event network with such a probability. Then, each active node receives a degree sampled from a different degree distribution with a longer tail (lognormal with
mean 10 and 99 percentile 100, cut off at 1000), to model super-spreading effects.\footnote{Note that, as we rely on distributions with cut-off, there is no practical difference in using a lognormal or a heavier tail distribution like a power law.}

In order to simulate our discrete-time model, with step corresponding to one day, we further give each agent one of four different states (\textbf{S}usceptible, \textbf{E}xposed but not infective, \textbf{I}nfective, \textbf{R}ecovered), and sample the duration in days of the transitions from E to I and from I to R according to a gamma distribution taken from COVID19 literature \cite{merler_2020}. Infection of a Susceptible node can happen if it is in contact with an Infective node, with a probability $p$ which is edge dependent and sampled for each edge according to a Beta distribution with a mean 0.05 (which roughly gives an R0 close to 2, as observed in the second phase of the COVID epidemics, in Lombardia, Italy). We assume independence among different edges when modeling  the spreading of infection.

At each time step, the spatial model of the epidemic network is designed by the pair  $\langle L, \wfun\rangle$, where the set of locations $L$ corresponds to the set of agents and the proximity function $\wfun$ is such that  $(\ell_i,w,\ell_j)\in \wfun$ if and only there is a probability greater than zero that $\ell_i$ and $\ell_j$
are in contact. 
The value of the weight $w$ corresponds to the sampled probability described above, both for the static and the dynamic one. 
More specifically, $w=-\ln(p_{\ell_i,\ell_j}(t))$,  where $p_{\ell_i,\ell_j}(t)$ is the infection probability at time t. 
Hence, the higher is $w$, the lower is the probability that agent $\ell_i$ is infected by agent $\ell_j$.  We consider two types of distances here, the $hop$ distance, counting the number of hops, and the $weight$ distance, summing the value of edges $w$.

The spatio-temporal trace of our epidemic network is $x: L \rightarrow \mathbb{T} \rightarrow \mathbb{Z}$ with only one signal $\vec x(\ell)= \tsign$ associating with each agent $\ell$ and at each time t, the state $x(\ell, t) \in \mathbb{S} = \{\textbf{Susceptible}, \textbf{Exposed}, \textbf{Infected}, \textbf{Recovered} \}=\{\textbf{S}, \textbf{E}, \textbf{I}, \textbf{R} \}$.
To give an idea of the behavior of this model we plot in Figure~\ref{fig:simulation} the number of nodes in each state at each time for a random simulation.

\begin{figure}[H]
    \centering
    \includegraphics[scale=0.6]{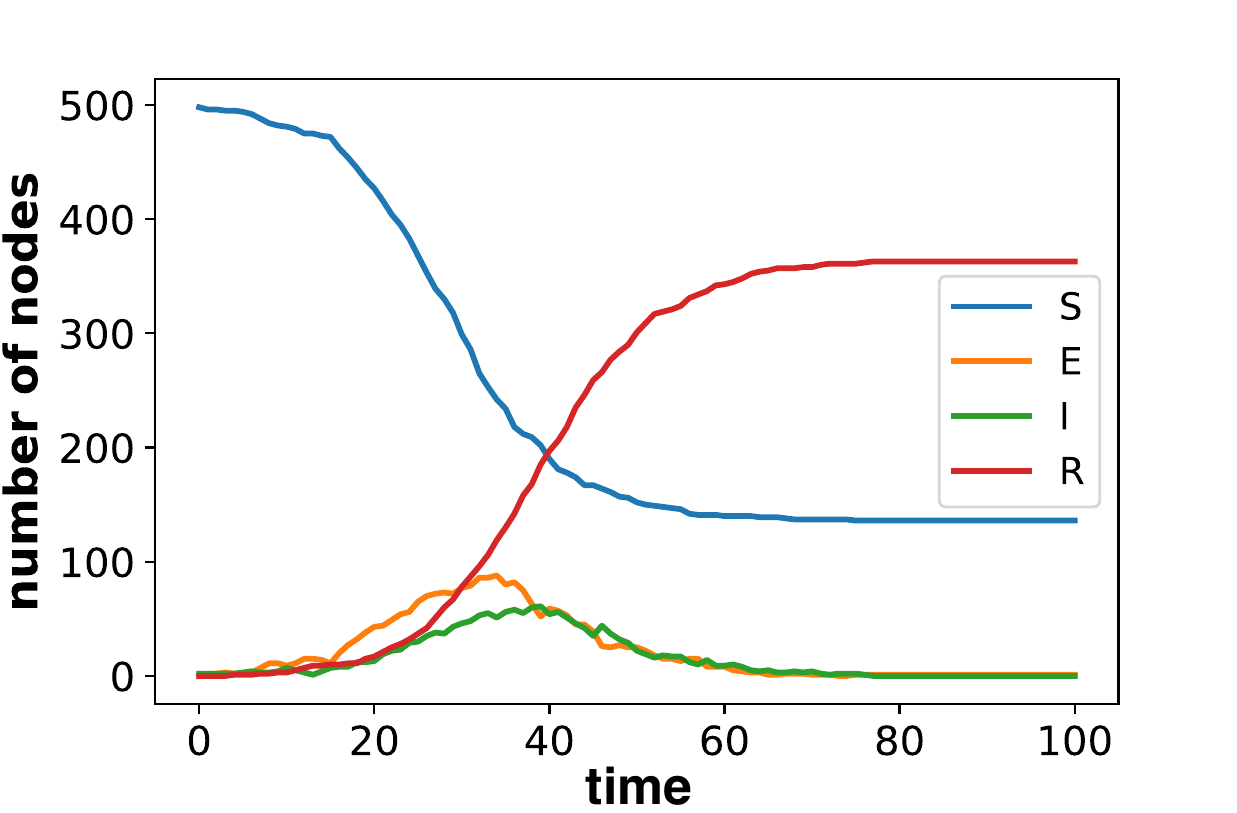}
    \caption{The number of nodes in each state at each time for a random simulation.}
    \label{fig:simulation}
\end{figure}

The first property that we consider is:
\begin{equation}
\phi_{dangerous\_days} = \glob{} (\textbf{Susceptible} \reach{[0,1]}{hop} (\ev{[0,2]}(\textbf{Infected})) => \ev{[0,7]} \textbf{Infected} )
\end{equation}

$\phi_{dangerous\_days}$ is satisfied in a location when it is always true (at each time step) that if a susceptible individual is directly connected with an individual that will be eventually infected in the next 2 days then it will eventually be infected within the next $7$ days.  If we consider only the static network, this property is satisfied on $447\pm12.5$ nodes of the network, considering 500 experiments, instead, considering only the dynamic network the property is satisfied by $350\pm70.5$ nodes. As expected the daily contacts are more dangerous than casual contact, and the dynamic network has more variability than the static one.

The second property that we consider is:
\begin{equation}
\phi_{safe} = \glob{} ( \everywhere{[0,r]}{weight}  (\neg \textbf{Infected}) => \glob{[0,T]}  (\neg \textbf{Infected})  )
\end{equation}

$\phi_{safe}$ holds in a location if it is always true that, when all the connected locations at a weight distance less than $r$ (i.e. with infection probability $\leq 10^{-r}$) are not infected, implies that this location will remain not infected for the next T days. With this property, we can study the relationship between the probability of being infected from connected nodes and being actually an infected individual after a number of days. If a location satisfies the property, it means that being in a radius $r$ of not infected individuals prevents infection.
If a location does not satisfy the property it means that there is some infected node at a distance of more than $r$, connected with it that causes its infection within the next T days. Setting $T=7$ days, we study the variation of $r$ versus the number of nodes that satisfy the property (in a network with 500 nodes). Figure~\ref{fig:safe_radius} shows the results. We report also a second scale with the corresponding probability value. We can see that
with $r=3$ which corresponds to a connection probability equal to $0.05$ (the mean of our Beta distribution), only half of the nodes satisfy the property, and to have almost all nodes that satisfy the property we need to consider a very large radius. This means that having in the network edges with very large values of $r$ will not prevent the spread of the disease.
\begin{figure}[H]
    \centering
    \includegraphics[scale=0.65]{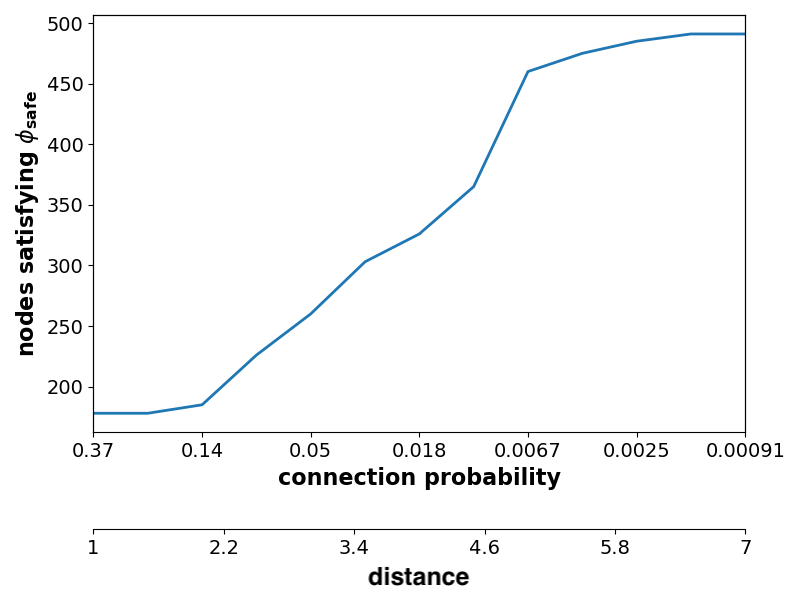}
    \caption{Number of nodes that satisfy property $\phi_{safe}$ versus parameter r.}
    \label{fig:safe_radius}
\end{figure}

\section{Conclusions}
\label{sec:conclusion}

We presented STREL, a formal specification language to express and to monitor spatio-temporal requirements over a dynamic network of spatially-distributed CPS.  Our monitoring framework considers the CPS topological configuration as a weighted graph with the nodes modeling the CPS entities while the edges representing their arrangement according to predefined spatial relations (e.g. proximity relation, connectivity, Euclidean distance, etc.). Both nodes and edges contain attributes that model physical and logical quantities that can change over time. 
STREL combines the Signal Temporal Logic~\cite{Maler2004} with two spatial operators \emph{reach} and \emph{escape} that are interpreted over 
the weighted graph. Other spatial modalities such as \emph{everywhere}, \emph{somewhere} and \emph{surround} can also be derived from them.
We demonstrated how STREL can be interpreted according to different semantics (Boolean, real-valued), defining a unified framework capturing all of them,  based on constraint semirings.  We  provided a generic offline monitoring algorithm based on such semiring formulation, providing also correctness  proofs  and  discussing  in  detail  its  algorithmic complexity.  We showed several examples of  requirements that we can monitor in our framework, considering two different case studies. 

As future works,  first, we aim to design a distributed and \emph{online} monitoring procedure where the spatio-temporal signal is not known at the beginning, and it is discovered while data are collected from the system; some preliminary results on this line can be found~\cite{strelonline}. Second, we aim to extend our framework with new features such as the possibility to synthesize automatically spatio-temporal controllers from the STREL specification or to provide automatically an explanation of the failure, enabling us to detect the responsible components when a STREL requirement is violated.

\section*{Acknowledgment}
This research has been partially supported by the Austrian FWF projects ZK-35 and W1255-N23, by the Italian PRIN project ``SEDUCE'' n.~2017TWRCNB, by the Italian PRIN project ``IT-MaTTerS'' n.~2017FTXR7S, and by POR MARCHE FESR 2014-2020, project ``MIRACLE'', CUP B28I19000330007.

\bibliographystyle{alphaurl}
\bibliography{biblio}

\end{document}